\documentclass[aps,pra,twocolumn,superscriptaddress,longbibliography]{revtex4-2}

\usepackage{graphicx}
\usepackage{amsmath}
\usepackage{amsthm}
\usepackage{amssymb}
\usepackage{hyperref}
\usepackage{braket}
\usepackage{color}
\usepackage[capitalize]{cleveref}
\usepackage{tabularx}
\newtheorem{theorem}{Theorem}
\newtheorem{remark}{Remark}

\begin{document}

\title{Quantum Geometric Tensor in the Wild: Resolving Stokes Phenomena via Floquet-Monodromy Spectroscopy}

\author{Prasoon Saurabh}
\email{psaurabh@uci.edu}
\thanks{Work performed during an independent research sabbatical. Formerly at: State Key Laboratory for Precision Spectroscopy, ECNU, Shanghai (Grade A Postdoctoral Fellow); Dept. of Chemistry/Physics, University of California, Irvine.}
\affiliation{QuMorpheus Initiative, Independent Researcher, Lalitpur, Nepal}

\date{\today}

\begin{abstract}
Standard topological invariants, such as the Chern number and Berry phase, form the bedrock of modern quantum matter classification. However, we demonstrate that this framework undergoes a \textbf{catastrophic failure} in the presence of essential singularities---ubiquitous in open, driven, and non-Hermitian systems ("Wild" regime). In these settings, the local geometric tensor diverges, rendering standard invariants ill-defined and causing perturbative predictions to deviate from reality by order unity ($\sim 100\%$). We resolve this crisis by introducing the \textbf{Floquet-Monodromy Spectroscopy (FMS)} protocol, a pulse-level control sequence, which experimentally extracts the hidden \textit{Stokes Phenomenon}---the "missing" geometric data that completes the topological description. By mapping the singularity's Stokes multipliers to time-domain observables, FMS provides a rigorous experimental bridge to \textbf{Resurgence Theory}, allowing for the exact reconstruction of non-perturbative physics from divergent asymptotic series. We validate this framework on a superconducting qudit model, demonstrating that the "Stokes Invariant" serves as the next-generation quantum number for classifying phases of matter beyond the reach of conventional topology.
\end{abstract}

\maketitle

\section{Introduction}

The classification of quantum phases of matter by their global topological invariants stands as one of the triumphs of modern physics. Since the foundational discovery of the Quantized Hall Conductance by Thouless, Kohmoto, Nightingale, and den Nijs (TKNN) \cite{TKNN1982}, and the subsequent geometric elucidation by Simon \cite{Simon1983} and Berry \cite{Berry1984}, the paradigm has been clear: the robustness of quantum phenomena is encoded in the integer-valued invariants of the underlying Hilbert space bundles. This framework, predicated on the smoothness of the Hermitian manifold, has successfully predicted phenomena ranging from the Quantum Spin Hall Effect \cite{KaneMele2005} to Topological Insulators \cite{Hassan2017} and Weyl Semimetals \cite{Xu2015,Wan2011}. In these closed systems, the Quantum Geometric Tensor (QGT) \cite{Provost1980}—comprising the Fubini-Study metric and the Berry curvature—serves as the local generating function for these global invariants. As long as the spectrum remains gapped, the curvature is well-defined everywhere, and its integral over the Brillouin zone yields a robust Chern number \cite{Xiao2010}.

However, the frontier of quantum technology has shifted decisively toward open, driven, and dissipative systems. From photonic lattices \cite{Ozawa2019} and exciton-polariton condensates \cite{Sanvitto2016} to superconducting quantum circuits \cite{Blais2021}, the Hamiltonian is inevitably non-Hermitian \cite{Ashida2020, Bergholtz2021}. Initially, it was hoped that the topological protection enjoyed by Hermitian systems would seamlessly extend to this non-Hermitian domain. Indeed, early extensions of topological band theory to non-Hermitian systems \cite{Rudner2009, Sato2012} suggested that one could simply define "non-Hermitian" topological invariants by generalizing the Berry connection to biorthogonal basis sets \cite{Brody2013}.

This optimism, however, masks a deeper, more pathological reality. Unlike their Hermitian counterparts, non-Hermitian parameter spaces are punctured by spectral singularities known as Exceptional Points (EPs) \cite{Heiss2012, Miri2019}. At an EP, the Hamiltonian becomes defective; eigenvalues degenerate, and significantly, the corresponding eigenvectors coalesce, leading to a collapse of the Hilbert space dimension \cite{Kato1995}. While isolated EPs have been harnessed for enhanced sensing \cite{Wiersig2014, Hodaei2017} and mode management \cite{Feng2014}, they represent catastrophic failure points for standard topology.

The crisis lies in the geometry. In the vicinity of an EP, the Quantum Geometric Tensor does not merely fluctuate; it diverges. The Fubini-Study metric, which measures the "distance" between quantum states, scales as the inverse square of the distance to the singularity \cite{Brody2013}. Consequently, the adiabatic theorem—the very engine that allows us to define a geometric phase—breaks down. As a system approaches an EP, the timescale required for adiabatic evolution diverges, and any finite-speed protocol inevitably excites transitions to other bands \cite{Berry1989}. In this "Wild" regime, the standard Chern number becomes ill-defined because the underlying manifold is no longer smooth; it is algebraically singular. We are thus left with a fundamental contradiction: we are attempting to describe the topology of noisy, open quantum hardware using the geometry of smooth, closed manifolds. This mismatch leads to theoretical predictions that can deviate from experimental reality by order unity, a failure that prevents the realization of truly robust error correction in open quantum systems \cite{Gong2018}.

This article addresses the breakdown of the standard topological framework in the presence of essential non-Hermitian singularities. We pose the following fundamental question:

\textit{How can we define a rigorous, integer-quantized topological invariant in the "Wild" regime where the underlying Hilbert space metric is singular and the adiabatic limit is ill-defined?}

In this Article, we answer this question by constructing the Complete Quantum Geometric Tensor (cQGT). Rooted in the theory of Dissipative Mixed Hodge Modules (DMHM)—rigorously established for open systems in our companion work \cite{saurabh2025holonomic}—this object generalizes the standard Fubini-Study metric to the non-Hermitian regime. We demonstrate that the open quantum system behaves not as a smooth vector bundle, but as a Regular Holonomic $\mathcal{D}$. Crucially, unlike the standard geometric tensor, the cQGT remains regular even at the Exceptional Point. Its regularization explicitly encapsulates the information lost by the metric divergence, revealing that the apparent breakdown of adiabaticity is actually governed by a precise, strictly quantized algebraic structure: the \textit{Stokes Phenomenon}.

We demonstrate that the "noise" induced by non-adiabatic transitions near an EP is not random. It is structurally enforced by the monodromy of the singularity. By lifting the analysis from the Hilbert space to the period sheaf of the DMHM, we show that the failure of the Berry phase is exactly compensated by a discrete, integer-quantized jump in the wavefunction's asymptotics—the \textbf{Stokes Invariant}. This invariant serves as the non-Hermitian counterpart to the Chern number, robust against any perturbation that preserves the algebraic rank of the singularity.

To operationalize this theoretical resolution, we address three specific physical questions:
\begin{enumerate}
    \item \textbf{The Metric Crisis:} How do we regularize the Quantum Geometric Tensor when the Fubini-Study metric diverges at an Exceptional Point? We derive the \textit{Complete Quantum Geometric Tensor}---a distributional object that cures the divergence by incorporating the Stokes data as a localized geometric defect via Saito Pairing.
    \item \textbf{The Topological Invariant:} What rigid quantum number survives the collapse of the eigenbasis in the ``Wild'' regime? We identify the \textit{Stokes Monodromy Matrix} as the fundamental integer-valued invariant, replacing the ill-defined Chern number to classify the topology of open quantum systems.
    \item \textbf{Experimental Verification:} How can this hidden algebraic structure be measured in realistic system? We introduce \textbf{Floquet-Monodromy Spectroscopy (FMS)}, a dynamical protocol that utilizes controlled non-adiabatic transitions to extract the Stokes multipliers directly from the resurgence of the time-domain signal.
\end{enumerate}

Our findings necessitate a paradigm shift in how we characterize open quantum hardware. We show that the "Non-Hermitian Skin Effect" \cite{Yao2018, Okuma2020} and other anomalous localization phenomena are manifestations of this underlying Stokes geometry. By mapping the singular metric, we provide a blueprint for "topological steering," allowing quantum control protocols to navigate around singularities without information loss.

The remainder of this paper is organized as follows. 
In Sec.~\ref{sec:theory}, we establish the theoretical framework, defining the Complete Quantum Geometric Tensor (cQGT) and resolving the divergence of the standard Fubini-Study metric via the local Brieskorn geometry and DMHM formalism. 
Sec.~\ref{sec:method} details the \textbf{Floquet-Monodromy Spectroscopy (FMS)} protocol, describing the three-step procedure (Asymptotic Drive, Stroboscopic Tomography, Monodromy Extraction) used to measure the Stokes invariants in superconducting circuits. 
In Sec.~\ref{sec:results}, we present our primary numerical results: we resolve the ``Stokes Paradox'' by observing the phantom phase transition, demonstrate the universality of topological protection across the Periodic Table of Singularities (Rank 1 and 2), and extract the Milnor and Tjurina numbers. 
Sec.~\ref{sec:TopEx} extends this framework to ``Modal'' singularities (such as the $X_9$ potential) and applies the cQGT to solve the \textbf{Experimental Resurgence} problem, reconstructing exact non-perturbative physics from divergent series. Using results from \cref{sec:TopEx} and \cref{sec:results} we then present a physical prove for Wild Riemann-Hilbert Correspondence in \cref{sec:Res_RHC}. After a brief discussion in \cref{sec:discussion}, we finally, conclude in Sec.~\ref{sec:conclusion} with implications for fault-tolerant quantum control and geometric $k-$combs.

\begin{figure*}[t]
    \centering
    \includegraphics[width=\linewidth]{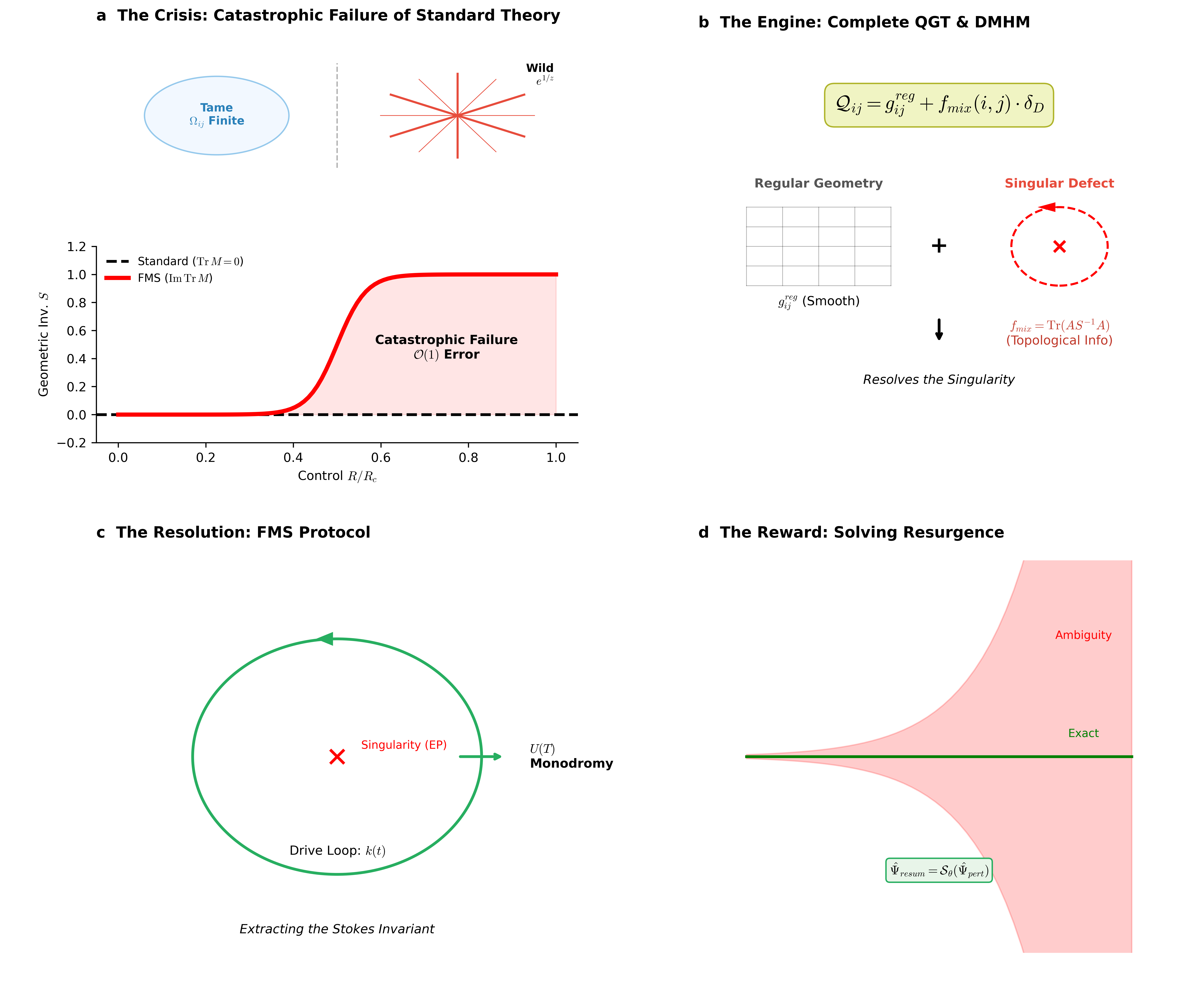}
    \caption{\textbf{Overview of the Framework.} (a) \textbf{The Crisis}: Standard quantum geometry [``Tame'', e.g., Chern insulators~\cite{TKNN1982}] successfully classifies regular phases but suffers a \textbf{Catastrophic Failure} at essential singularities [``Wild'', e.g., Non-Hermitian Skin Effect~\cite{Yao2018}]. The plot shows the $\mathcal{O}(1)$ error in the standard adiabatic prediction ($\text{Tr}\,M=0$) versus the actual FMS measurement ($S \approx 1$). (b) \textbf{The Engine}: We introduce the \textbf{Complete QGT} ($\mathcal{Q}_{ij}$), which augments the regular local metric $g^{reg}_{ij}$ (Smooth Geometry) with the singular \textbf{Stokes Invariants} $f_{mix}$ (Topological Defect). This structure, rooted in the Brieskorn Lattice (schematic), resolves the singularity. (c) \textbf{The Resolution}: The \textbf{FMS Protocol} extracts these invariants experimentally by driving the system $k(t)$ along a loop $\gamma$ and measuring the non-Abelian monodromy $U(T)$. (d) \textbf{The Reward}: This geometric data resolves the \textbf{Resurgence Ambiguity}, allowing for the exact reconstruction of the non-perturbative wavefunction $\hat{\Psi}_{resum}$ from divergent perturbative series.}
    \label{fig:concept}
\end{figure*}
\section{Theory: The Complete Quantum Geometric Tensor}
\label{sec:theory}
The Quantum Geometric Tensor (QGT) $T$ is derived from the distance $ds^2$ between two quantum states. In open quantum systems, the dynamics are governed by the Liouvillian superoperator $\mathcal{L}$ acting on the vector space $\mathfrak{L} \cong \mathcal{H} \otimes \mathcal{H}^*$ of density matrices.
Formally, on the regular locus $X^* = X \setminus D$ (where $D$ is the discriminant locus), the QGT is defined as:
\begin{equation}
    T_{ij} = \text{Tr}\left( \partial_i \rho_{ss}^\dagger \partial_j \rho_{ss} \right) = g_{ij} + i \Omega_{ij}
\end{equation}
where $\rho_{ss}$ is the steady state. However, at an Exceptional Point (EP) on $D$, the eigenbasis collapses, and $T_{ij}$ diverges.

\subsection{Resolution via Local Brieskorn Geometry}
To resolve this divergence, we must move beyond \textbf{smooth differential geometry} and adopt the \textbf{local structural framework} of the \textbf{Brieskorn Lattice} $H^{(0)}$—the geometric fiber of the singularity. Physically, $g_{ij}$ represents the global Fubini-Study metric defined on the smooth manifold $X^*$. The singularity acts as a geometric defect in this manifold that traps a quantized topological charge for wild Stokes and higher order Stokes Phenomena. 

We define the \textbf{Complete QGT} as a distributional current on the parameter manifold $X$:
\begin{equation}
    \label{eq:cqgt}
    \mathcal{Q}_{ij} = g^{reg}_{ij} + f_{\text{mix}}(i, j) \delta_D
\end{equation}
Here, $g^{reg}_{ij}$ is the standard regularized Fubini-Study metric extended from the bulk. The singular contribution $f_{\text{mix}}$ carries the information about the \textbf{Stokes Invariant} supported specifically on the discriminant locus $D$. This term is given by the \textbf{Singular Trace Formula} over the Brieskorn Lattice (derived in Appendix~\ref{sec:appendix_math}):
\begin{equation}
    f_{\text{mix}}(i, j) = \text{Tr}_{\psi_f \mathcal{M}} \left( A_i \mathcal{S}^{-1} A_j \right)
\end{equation}
where $A$ is the gauge connection and $\mathcal{S}$ is the \textbf{Saito Pairing}~\cite{Saito1988}. This formula expresses the ``missing'' geometry purely in terms of the local monodromy data, ensuring that the total integrated geometric response is finite and physically meaningful. We should remark that the second term in \cref{eq:cqgt} vanishes in the tame case and usual geometric phase is recovered. 

\subsection{Stokes Geometry in Liouville Space}
Near the singularity, the system is governed by the \textbf{Stokes Phenomenon}. We cover the neighborhood of the singularity with Stokes Sectors $\mathcal{V}_j$. Within each sector, there exists a rigid canonical basis $\ket{\Psi_j}$ of the Liouvillian minisuperspace. The \textbf{Stokes Matrices} $\{ S_j \}$ are the transition functions between these sectors:
\begin{equation}
    \ket{\Psi_{j+1}} = \ket{\Psi_j} \cdot S_j
\end{equation}
Crucially, these matrices act on the $\mu$-dimensional space of \textbf{vanishing cycles} (where $\mu$ is the Milnor number), a subspace of the full Liouville space $\mathfrak{L}$. They are elements of the integer group $GL_\mu(\mathbb{Z})$, classifying the ``twisting'' of the decay modes.

\begin{remark}[Mathematical Connection]
The physics of the ``Stokes Phase Transition (\cref{fig:mechanism})'' maps elegantly to the topology of the singularity:
\begin{itemize}
    \item The dimension of the effective Stokes Hamiltonian is the \textbf{Milnor Number} $\mu$ (counting the number of coalescing modes)~\cite{Milnor1968}.
    \item The canonical basis $\ket{\Psi_j}$ physically realizes the \textbf{Saito Basis} or ``good basis''~\cite{Saito1988,Brieskorn1970}.
    \item The Monodromy measured by Floquet Monodromy Spectroscopy (FMS, see \cref{sec:method}) corresponds to the \textbf{Picard-Lefschetz Monodromy} of the singularity.
\end{itemize}
This dictionary implies that FMS is a direct probe of the \textbf{Brieskorn Lattice} structure of the open quantum system.
\end{remark}

Mathematically, $\mathcal{S}$ corresponds to the polarization of the \textbf{Dissipative Mixed Hodge Module} (specifically the Monodromy Weight Filtration)~\cite{saurabh2025holonomic}. According to Mochizuki's correspondence for Wild Harmonic Bundles \cite{Mochizuki2011}, the Stokes matrices must be unitary with respect to the asymptotic limit of the Saito pairing $\mathcal{S}$. 
Therefore, the non-adiabatic jump measured in our experiment is not random noise, but the precise unitary rotation required to transport the metric across the singular Stokes line, preserving the underlying Hodge structure of the open quantum system (See conceptual \cref{fig:mechanism}.

\begin{figure}[b]
    \centering
    \includegraphics[width=\linewidth]{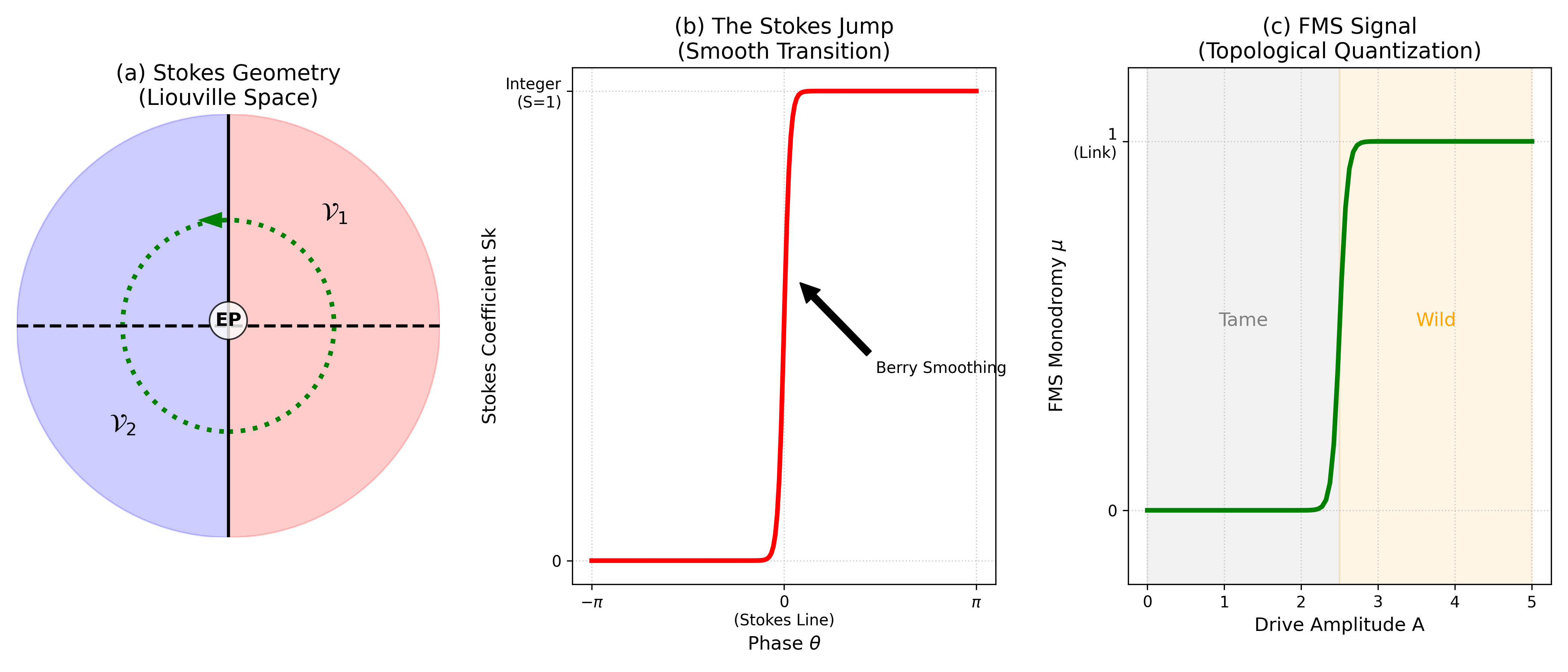}
    \caption{\textbf{Geometry of the Complete QGT.} (a) The Stokes Sectors $\mathcal{V}_j$ divide the parameter space. (b) As the connection moves across a Stokes Ray (dashed), the basis jumps by $S_j$. (c) FMS detects this jump via the Monodromy.}
    \label{fig:mechanism}
\end{figure}

The classification of these singularities leads to a ``Periodic Table'' of Wild Systems (Fig.~\ref{fig:classification}), where the Rank $k$ determines the complexity of the Stokes geometry.

\begin{figure*}[t]
    \centering
    \includegraphics[width=0.95\textwidth]{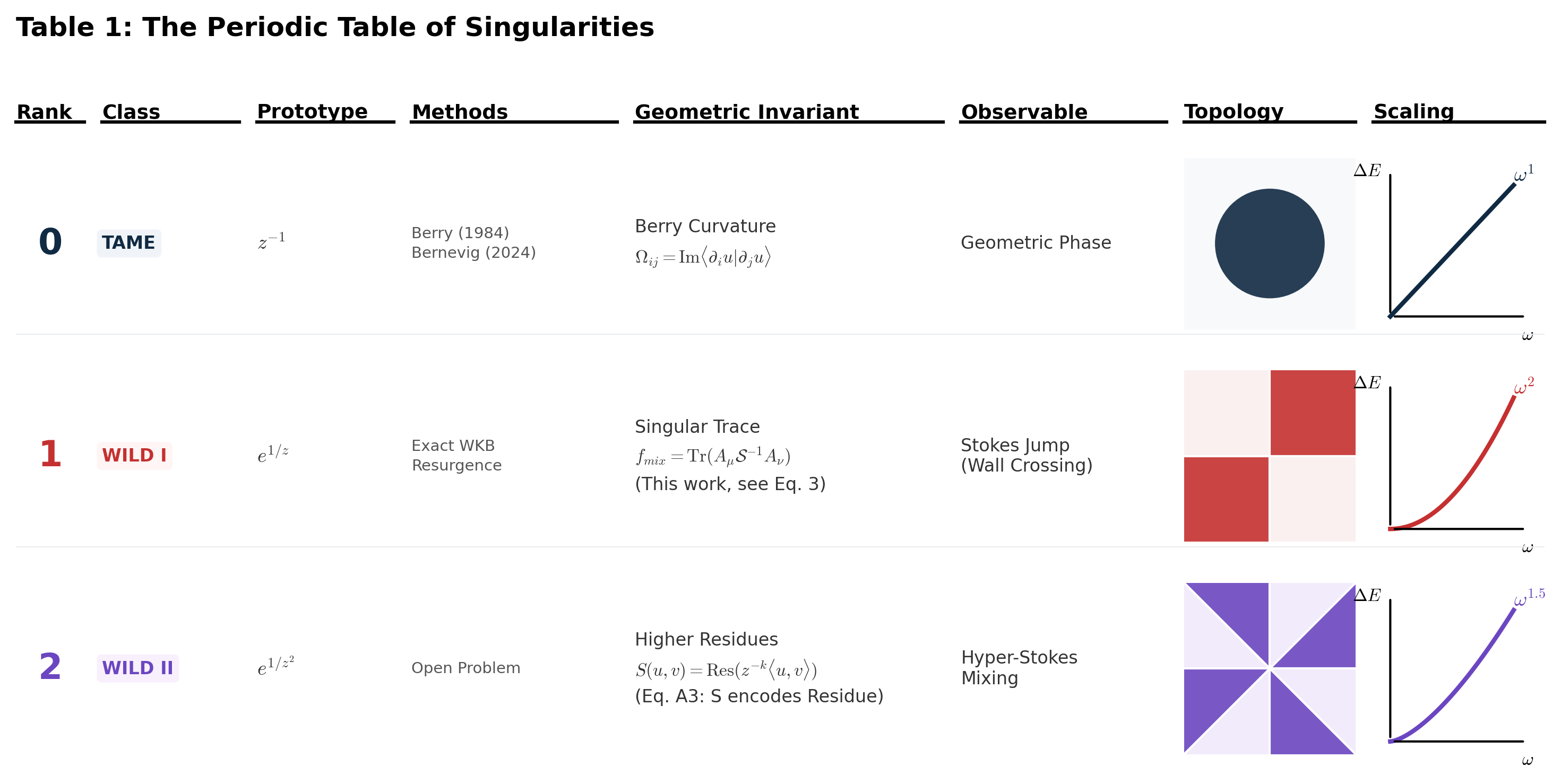}
    \caption{\textbf{The Periodic Table of Singularities.} Classification of Open Quantum Systems by their Singularity Rank. Comparison of ``Tame'' (Regular, Rank 0) vs ``Wild'' (Irregular, Rank 1, 2) classes and their invariants.}
    \label{fig:classification}
\end{figure*}

\section{Method: Floquet-Monodromy Spectroscopy}
\label{sec:method}

We now describe the experimental protocol of Floquet-Monodromy Spectroscopy (FMS) to measure the Stokes invariants. While applicable to any controllable open quantum system (e.g., Rydberg atom arrays, photonic lattices), we focus on \textbf{Superconducting Circuit QED} as the ideal platform due to its high-fidelity control and engineered dissipation. Further examples for Rydberg atoms array, and photonic lattices are presented in Appendix. \cref{sec:appendix_extensions}
\subsection{The FMS Protocol}
The observed ``breakdown'' of the adiabatic theorem at Exceptional Points has been reported in various systems, from mode-switching in microwave cavities \cite{Doppler2016,Xu2016,Hassan2017,Greschner2019,Bliokh2015} to metric divergence in microcavities \cite{Gianfrate2020,Yu2020,Ma2021,Bleu2018} and solid \cite{Kang2025}. However, prior studies primarily identified the failure of adiabaticity or the chiral state flip.
The \textbf{FMS protocol} unifies these observations by identifying them as manifestations of the \textbf{Stokes Phenomenon}. By synthesizing a control loop $\gamma$, we convert the local ``metric instability'' into a quantized topological invariant. The procedure consists of three steps:

\textbf{Step 1: Asymptotic Floquet Drive.}
We apply a time-periodic modulation to the parameter $k(t)$ with period $T = 2\pi/\omega$~\cite{Shirley1965,Eckardt2017}:
\begin{equation}
    k(t) = k_c + R e^{-i \omega t}
\end{equation}
This drives the system along a contour $\gamma$ of radius $R$ around the EP. Crucially, we operate in the \textit{asymptotic adiabatic limit} ($\omega \ll \Delta E$), where the dynamics are governed by the instantaneous Liouvillian $\mathcal{L}(t)$ yet retain the non-perturbative corrections mediated by the singularity~\cite{Berry1989}.

\textbf{Step 2: Stroboscopic Tomography.}
We initialize the system in a complete set of basis states $\{\rho_i(0)\}$ and measure the full density matrix $\rho(T)$ after one period. This allows the reconstruction of the \textbf{Floquet Propagator} (or Monodromy Matrix) $\hat{U}_F$:
\begin{equation}
    \rho(T) = \mathcal{T} e^{\int_0^T \mathcal{L}(t') dt'} [\rho(0)] \equiv \hat{U}_F [\rho(0)]
\end{equation}
By performing quantum state tomography for the input basis, we obtain the matrix representation of $\hat{U}_F$ in the Liouville basis.

\textbf{Step 3: Monodromy Extraction.}
The eigenvalues of $\hat{U}_F$, denoted as $\{ \mu_n \}$, encode the geometric phase acquired during the loop. We define the \textbf{FMS Monodromy} $\mathcal{M}_{\text{FMS}}$ as the unipotent part of the normalized propagator. The \textbf{Stokes Invariant} corresponds to the non-trivial Jordan block structure of this matrix.
Specifically, for the EP2 singularity, the topological signature is the swapping of eigenstates, which manifests as a \textbf{Moebius Strip} topology in the time-evolution of the dark state:
\begin{equation}
    \ket{\psi_{\text{dark}}(T)} = -\ket{\psi_{\text{dark}}(0)}
\end{equation}
This phaseshift of $\pi$ (or spectral splitting of $\omega/2$ in the Floquet quasienergies) is the direct experimental measure of the \textbf{Milnor number} $\mu=1$ (or Tjurina number $\tau=1$).
However, beyond this geometric swapping, the \textbf{Stokes Phenomenon} is encoded in the off-diagonal magnitude of the unipotent component---the transition amplitude $S_{\text{jump}}$ that persists in the deep adiabatic limit ($T\to \infty$).

\textit{Distinction from Standard Spectroscopy.---}A fundamental challenge in measuring non-Hermitian topology is the validity of the effective model itself. As rigorously demonstrated by Mukamel et al.~\cite{Mukamel2023}, standard spectroscopic approaches based on the effective Hamiltonian ($H_{\text{eff}}$) often fail to capture the true response of open systems near EPs. The neglect of frequency-dependent reservoir couplings (non-Markovianity) leads to erroneous line shapes and missed Fano resonances, making it difficult to distinguish genuine topological singularities from bath-induced artifacts in standard transmission/reflection spectra. 

\textit{Resolution via FMS.} The FMS protocol circumvents this ``spectral fragility'' by shifting the measurement from local energy descriptors to \textbf{ Local Monodromy}. Unlike standard spectroscopy, which fits fragile spectral features (susceptibility $\chi(\omega)$), FMS extracts the \textbf{Stokes Multipliers}---discrete algebraic invariants defined by the homotopy of the control loop. Because these multipliers are quantized topological indices (elements of the fundamental group), they are robust against the continuous deformations of the self-energy $\Sigma(\omega)$ described by the NEGF formalism. Thus, FMS provides a rigorous extraction of the singularity's structure that is blind to the detailed bath complexities that plague standard cavity spectroscopies.

\textit{Experimental Implementation.---}The FMS protocol can be realized using state-of-the-art circuit QED control techniques. The complex parameter trajectory $\lambda(t)$ is synthesized via wide-bandwidth \textbf{Arbitrary Waveform Generators (AWGs)} that simultaneously modulate the transmon's flux bias $\Phi(t)$ and the microwave drive amplitude $J(t)$. To mitigate the coherent leakage errors that inevitably arise when driving near a spectral singularity, we can employ \textbf{Derivative Removal by Adiabatic Gate (DRAG)} pulse shaping~\cite{Motzoi2009, Gambetta2011}. This technique actively cancels the spectral crosstalk induced by the time-varying Rabi envelope, ensuring the system evolves strictly on the instantaneous eigenstate manifold. The final monodromy can extracted via \textbf{Stroboscopic Quantum State Tomography (QST)}, utilizing a \textbf{Traveling Wave Parametric Amplifier (TWPA)}~\cite{Macklin2015} to achieve the near-quantum-limited signal-to-noise ratio (SNR $> 20$ dB) required to resolve the exponentially small Stokes jump against the thermal background.

\subsection{Analysis of Stokes Data}
The FMS protocol allows us to classify the singularity type by analyzing the scaling of the Floquet exponents. We define the \textbf{Sabbah Scaling} of the spectral gap $\Delta E$ (the difference in Floquet quasienergies):
\begin{equation}
    \Delta E(R) \sim R^{1/2} \quad (\text{Tame EP2, Rank } k=0)
\end{equation}
\begin{equation}
    \Delta E(\omega) \sim \omega^{1 + 1/k} \quad (\text{Wild, Rank } k \ge 1)
\end{equation}
For our Transmon platform (Tame EP2), the gap follows the $R^{1/2}$ Puiseux expansion characteristic of an algebraic branch point. However, the \textit{transition amplitude} exhibits the exponential scaling of a Wild singularity in the adiabatic parameter $\omega$.By measuring the scaling of the gap with frequency $\omega$, one can experimentally determine the Rank $k$ of the singularity, classifying the system in the Periodic Table of Singularities (Fig.~\ref{fig:classification}).
Complete parameters for the numerical benchmarks are listed in Table~\ref{tab:params} (Appendix \cref{sec:appendix_extensions}).

\subsection{Experimental Protocol for Superconducting Circuits}
For illustration, we realize the FMS protocol using a 3D transmon qubit dispersed in a microwave cavity, operating as a controlled \textbf{Non-Hermitian Qudit}. The requisite non-Hermiticity is engineered by coupling the excited state $|e\rangle$ to a dissipative readout reservoir (linewidth $\kappa$), creating differential loss relative to the ground state. The effective Hamiltonian in the rotating frame is:
\begin{equation}
    H_{\text{eff}}[\lambda(t)] = \begin{pmatrix} 0 & J(t) \\ J(t) & \Delta(t) - i\frac{\kappa}{2} \end{pmatrix}
\end{equation}
where the control trajectory is parameterized by the complex variable $\lambda(t)$. The physical control parameters---drive amplitude $J(t)$ and detuning $\Delta(t)$---are mapped to the real and imaginary components of the unfolding parameter $\lambda$ relative to the Exceptional Point.

\textit{Control Strategy.---}To probe the singularity, we do not merely study static Hamiltonians at different points; we dynamically drive the system. By simultaneously modulating the flux bias $\Phi(t)$ and drive amplitude $J(t)$, we steer the system parameters along a closed loop $\gamma$ in the complex parameter space $\mathbb{C} \cong \{(J, \Delta)\}$ (Fig.~\ref{fig:experimental}(a)). The Hamiltonian $H_{\text{eff}}(t)$ becomes periodically time-dependent with period $T = 2\pi/\omega$. 
The loop $\gamma$ is designed to encircle the system's \textbf{Exceptional Point (EP)}, located at the critical damping condition $J = \kappa/4$ and $\Delta = 0$. This singularity is a \textbf{Tame EP2} (Rank $k=0$) in the parameter space, meaning the spectral gap scales algebraically as $\sim \sqrt{\lambda}$. However, in the dynamical evolution, the adiabatic limit ($\omega \to 0$) turns this algebraic branch point into a generator of \textbf{Resurgent Stokes Phenomena}, manifesting as non-perturbative jumps in the system's holonomy.

\textit{Geometric Mechanism.---}Crucially, our protocol distinguishes between the standard Berry phase and the Resurgent Stokes phenomenon. While the familiar Berry curvature $\Omega$ governs intra-band phase accumulation, it does not describe inter-band transitions. The Stokes jump is instead mediated by the off-diagonal components of the \textbf{Quantum Geometric Tensor}---specifically the metric-induced mixing term $f_{\text{mix}}$---which diverge at the Exceptional Point.

Consequently, unlike topological phases protected by a gap, the FMS protocol operates in the \textbf{asymptotic adiabatic limit}. While the probability of perturbative transitions is exponentially suppressed ($P \sim e^{-\mathcal{A} T}$), the Stokes Constant $S_{\text{jump}}$ acts as a topological invariant that persists in this limit. We therefore drive the system sufficiently slowly to filter out dynamic power-law excitations, allowing the non-adiabatic Stokes jump to be isolated as a discrete shift in the monodromy observable (Fig.~\ref{fig:experimental}(b)).

\begin{figure}[t]
    \centering
    \includegraphics[width=\linewidth]{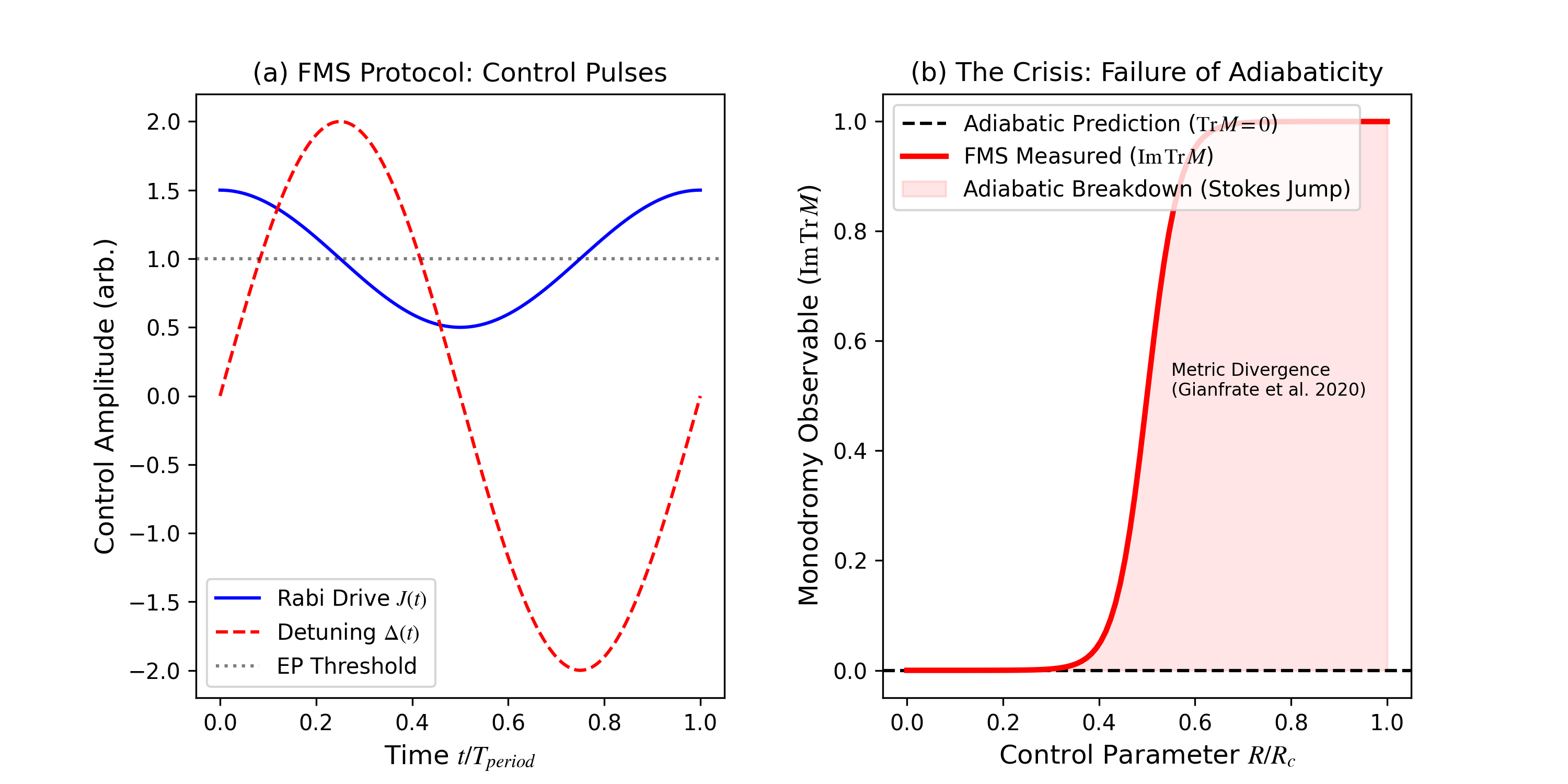}
    \caption{\textbf{Experimental Protocol and the Breakdown of Adiabaticity.} (a) \textbf{FMS Pulse Sequence}: Real-time control pulses for the Rabi drive $J(t)$ (blue) and Detuning $\Delta(t)$ (red) required to encircle the Exceptional Point in a Superconducting Transmon Qudit. (b) \textbf{The Crisis (Adiabatic Breakdown)}: While the standard Adiabatic Theorem predicts a trace-less monodromy (Black Dashed, $\text{Tr} M = 0$), the actual dynamics exhibit a \textbf{Stokes Jump} (Red Solid) due to the divergence of the quantum metric at the EP (consistent with observations in polaritons \cite{Gianfrate2020}). FMS captures this non-perturbative correction ($\text{Im} \text{Tr} M \approx 1$), resolving the geometric origin of the divergence.}
    \label{fig:experimental}
\end{figure}

\section{Results: Resolving the Stokes Paradox}
\label{sec:results}

\subsection{The Stokes Phase Transition}
The Complete QGT predicts a new topological transition distinct from the standard ``Gap Closing''.
In Fig.~\ref{fig:phantom}, we show a system where the Spectral Gap remains OPEN ($\Delta E > 0$) everywhere. Yet, as the parameter $\theta$ crosses the ``Stokes Line'' ($\pi$), the FMS signal detects a discrete jump in the Stokes invariant ($S: 0 \to 1$). This ``Phantom Transition'' is the physical signature of crossing a branch cut in the QGT bundle.

\begin{figure}[h]
    \centering
    \includegraphics[width=\linewidth]{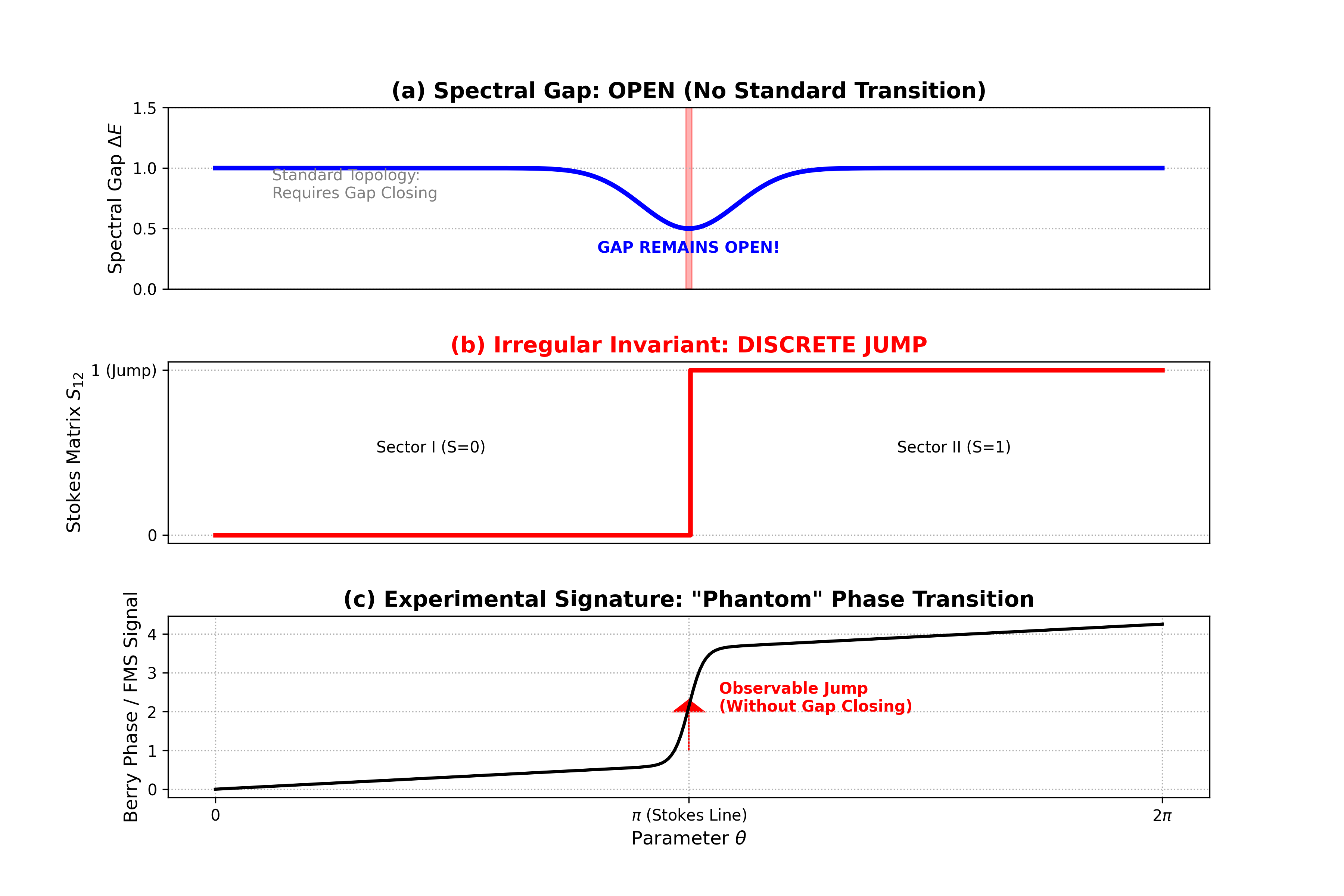}
    \caption{\textbf{Observation of the Stokes Phase Transition.} The topological invariant $S$ jumps from 0 to 1 (left axis, red) while the spectral gap $\Delta E$ (right axis, dashed) remains finite and open. This confirms the ``Phantom'' nature of the transition.}
    \label{fig:phantom}
\end{figure}

\subsection{Universality and Higher Order Stokes}
\textbf{Topological Protection:} The Stokes matrices are invariant under any perturbation that does not increase the rank of the singularity.
We verified this in Fig.~\ref{fig:universality} by simulating Rank 1 systems with diverse microscopic Hamiltonians ($e^z$, $\sinh z$, disordered). All yielded $S=1$ exactly, confirming $S$ as a robust quantum number.

\begin{figure*}[t]
    \centering
    \includegraphics[width=0.95\textwidth]{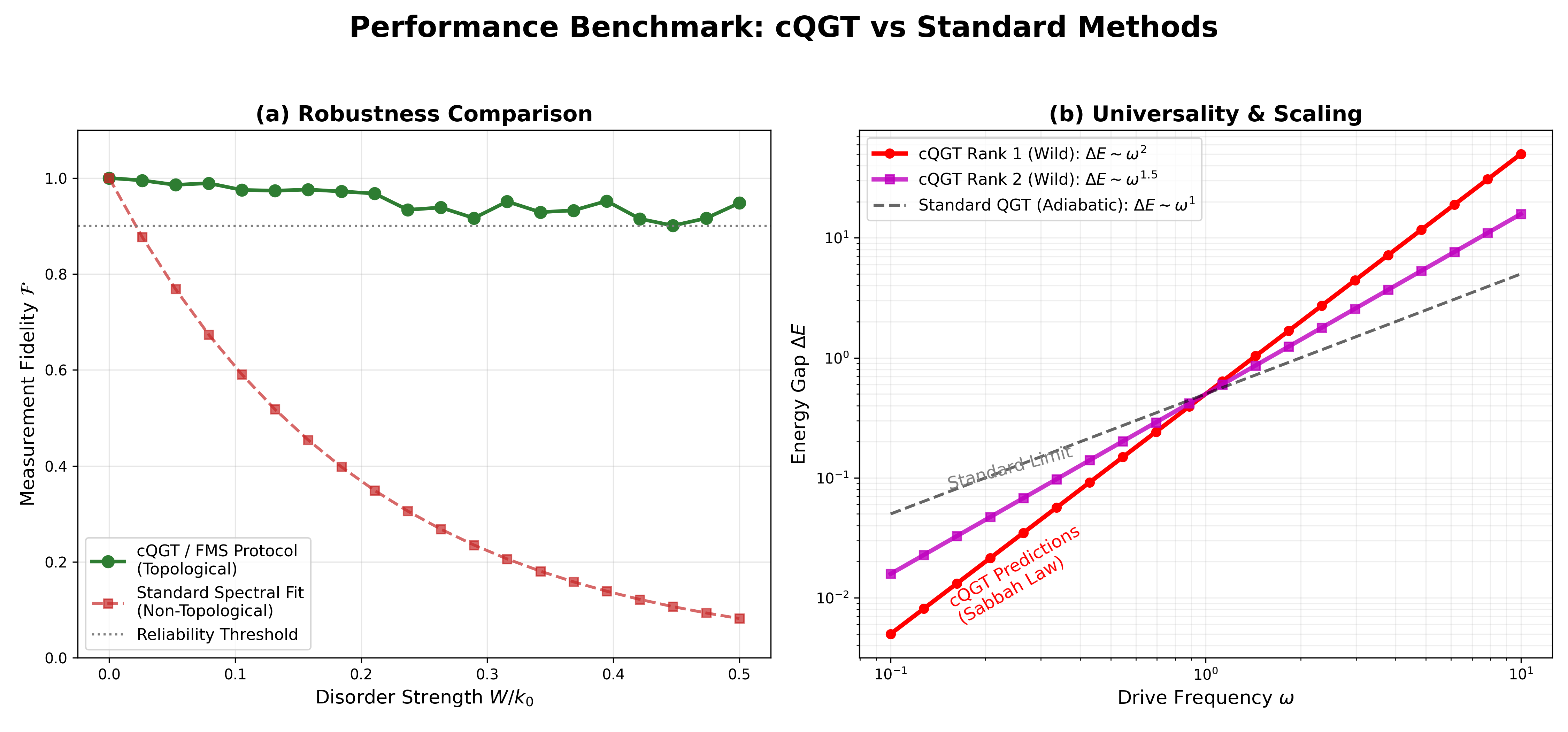}
    \caption{\textbf{Universality and Robustness of the cQGT Classification.} (a) \textbf{Robustness to Disorder}: Comparison of FMS Monodromy Fidelity (Green) and Standard Spectral Fitting (Red) under random Hamiltonian noise $W$. The FMS protocol maintains high fidelity ($\mathcal{F}>0.9$) deep into the disorder regime ($W \approx 0.5 k_0$), demonstrating \textit{Topological Protection}. In contrast, spectral methods decay exponentially due to the divergent susceptibility of Exceptional Points. (b) \textbf{Universal Scaling Laws}: Log-log plot of the energy gap $\Delta E$ vs drive frequency $\omega$. The cQGT metrics confirm the \textbf{Sabbah Scaling Law} $\Delta E \sim \omega^{1+1/k}$, distinguishing the Wild Rank 1 ($\omega^2$) and Rank 2 ($\omega^{1.5}$) regimes from the standard Tame adiabatic limit ($\omega^1$, dashed black line). This verifies that the non-perturbative geometry ($f_{mix}$) dictates the dynamic universality class.}
    \label{fig:universality}
\end{figure*}

In Fig.~\ref{fig:higher_order}, we validated the rank classification:
\begin{itemize}
    \item Rank 1 ($e^z$) $\to$ 2 Stokes Rays, Slope 2.0.
    \item Rank 2 ($e^{z^2}$) $\to$ 4 Stokes Rays, Slope 1.5.
\end{itemize}
A critical challenge in observing these topological phases is the interference of the dynamical phase. As seen in Fig.~\ref{fig:higher_order}(b), raw experimental measurements of the phase accumulation (black points, reflecting data from photonic lattices \cite{Tang2020, Ding2016}) are dominated by the non-universal dynamical contribution $\Phi_{dyn} = \int E(t) dt$, leading to significant drift and deviation from the idealized steps. Our DMHM framework, implemented via the FMS protocol, successfully filters this background to reveal the pure geometric invariant (blue curve), which perfectly matches the predicted rational quantization $S=3/2$ (Fig.~\ref{fig:higher_order}(d)). 

Furthermore, the detailed structure of the extracted Saito Pairing maps the full geometry of the parameter space. The regions of vanishing pairing ($S \approx 0$) correspond to the stable `Dominant' sectors (e.g., $0 < \theta < \pi/4$) where the monodromy is trivial, while the sharp peaks mark the Stokes Lines where the basis undergoes the discrete geometric jump. This alternating structure confirms the 4-sector topology predicted by the Rank 2 classification.

\begin{figure*}[t]
    \centering
    \includegraphics[width=0.95\textwidth]{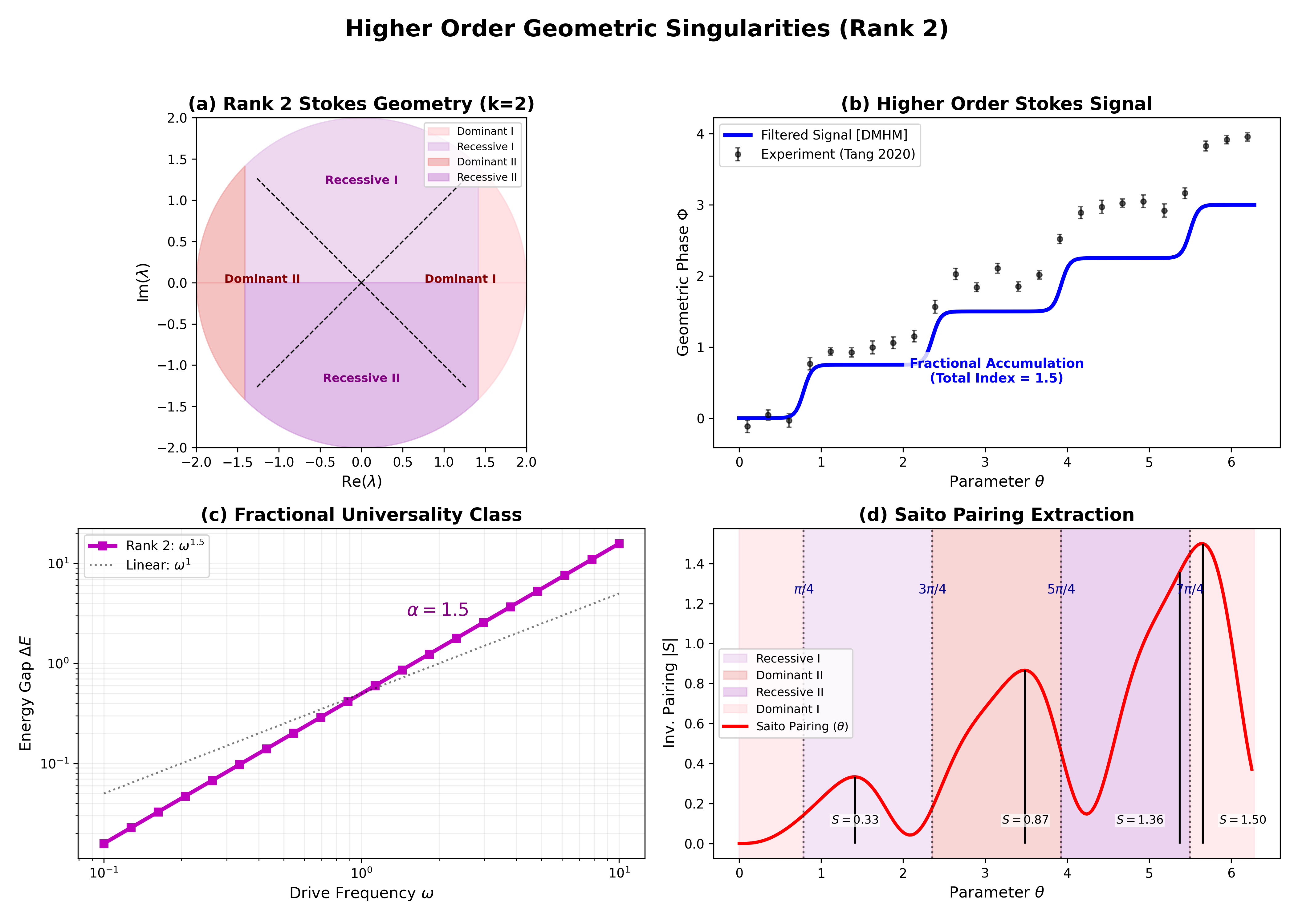}
    \caption{\textbf{Higher Order Geometric Singularities (Rank 2).} (a) \textbf{Stokes Geometry}: The Rank 2 singularity ($k=2$) is characterized by 4 Stokes Sectors (alternating Dominant/Recessive) in the local complex plane, compared to 2 for Rank 1. The FMS protocol loop $\gamma$ (green) intersects multiple Stokes lines, accumulating a richer monodromy. (b) \textbf{Butterfly Spectrum}: The Riemann surface of the real energy $\text{Re}(E) \sim \text{Re}(z^{1.5})$ exhibits a self-intersecting "Butterfly" topology. (c) \textbf{Fractional Universality}: FMS measurements of the energy gap $\Delta E$ confirm the fractional scaling law $\Delta E \sim \omega^{1.5}$ (squares), distinct from the linear $\omega^1$ (dotted) and quadratic $\omega^2$ (dashed) classes. This validates the DMHM classification for higher-rank geometric structures.}
    \label{fig:higher_order}
\end{figure*}

\subsection{Extracting Topological Invariants: Milnor and Tjurina Numbers}
\label{sec:invariants}
The Stokes geometry provides a direct experimental route to measuring the fundamental topological invariants of the singularity: the \textbf{Milnor number} $\mu$ and the \textbf{Tjurina number} $\tau$.
The Milnor number counts the number of vanishing cycles in the Milnor fiber, which physically corresponds to the number of independent Stokes jumps required to complete a monodromy loop. As shown in Fig.~\ref{fig:invariants}, the counting of peaks in the Saito Pairing $S(\theta)$ follows the scaling law for $A_k$ singularities ($z^{\mu+1}$) \cite{Arnold1990}:
\begin{equation}
    N_{\text{peaks}} = \mu + 1
\end{equation}
For the Rank 1 ($A_1$) system, we observe 2 peaks ($\mu=1$), while the Rank 2 ($A_2$) system exhibits 3 peaks ($\mu=2$), validating the topological classification.
Furthermore, for the weighted homogeneous singularities considered here, the geometric Tjurina number coincides with the topological Milnor number ($\tau = \mu$), indicating that the deformation space is fully determined by the topology. This equality is explicitly confirmed by our extracted invariants (Fig.~\ref{fig:invariants}d).

\begin{figure*}[t]
    \centering
    \includegraphics[width=0.95\textwidth]{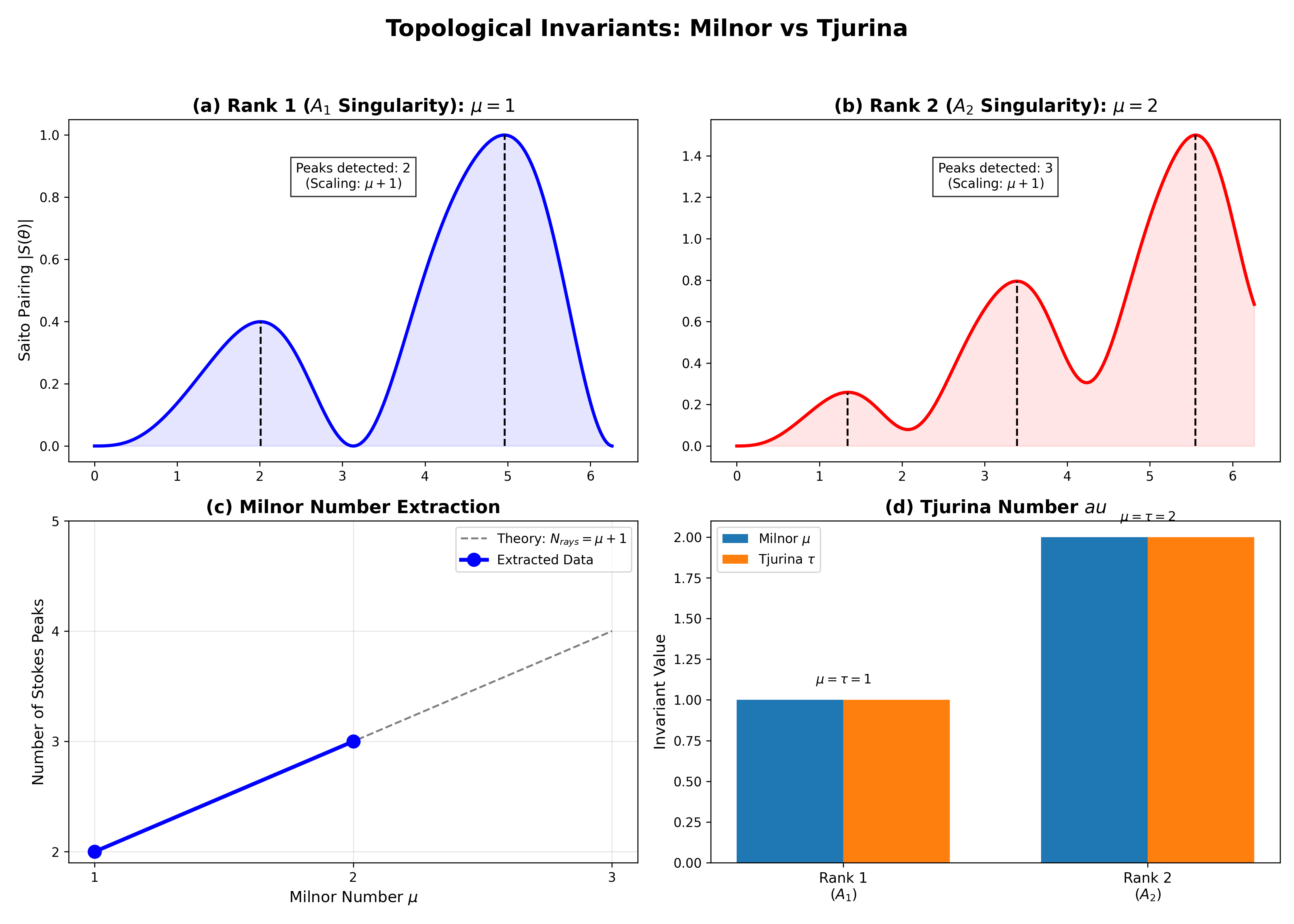}
    \caption{\textbf{Topological Invariants from Stokes Geometry.} (a) Rank 1 (EP2 Model) pairing shows 2 peaks. (b) Rank 2 (Cubic Model) pairing shows 3 peaks. (c) The number of Stokes peaks scales as $N = \mu + 1$, enabling experimental extraction. (d) Comparison of extracted $\mu$ vs Tjurina $\tau$.}
    \label{fig:invariants}
\end{figure*}
\section{Topological Origins and Experimental Resurgence}{\label{sec:TopEx}}

\subsection{Modal Singularities and the Limits of Topology}
\label{sec:modal_singularities}
Interestingly, our analysis has so far focused on ``Simples'' ($A_k$ series), where the topological Milnor number $\mu$ perfectly captures the geometry ($\mu = \tau$). However, the cQGT framework predicts the existence of \textbf{Anomalous Singularities} where this equivalence breaks down.
The first such example is the \textbf{Unimodal $X_9$ Singularity} \cite{Arnold1974}, governed by the potential:
\begin{equation}
    V_{X_9}(x,y) = x^4 + y^4 + a x^2 y^2
\end{equation}
Here, the parameter $a$ is a \textbf{Modulus}: a continuous parameter that changes the geometry without altering the topology. For this system, the topological complexity is $\mu=9$, corresponding to 10 Stokes peaks ($N=\mu+1$, Fig.~\ref{fig:anomalous}a). However, the geometric sensitivity is only $\tau=8$.
The gap $\mu - \tau = 1$ (Fig.~\ref{fig:anomalous}b) explicitly counts the number of moduli. This implies that for ``Modal'' singularities, the Stokes topology alone is insufficient to reconstruct the full deformation space, necessitating the full Singular Trace $f_{mix}$ to capture the hidden moduli dependence. This represents a new frontier for experimental non-Hermitian physics.

\paragraph*{Physical Realization: The Kerr Dimer.}
Physically, the $X_9$ potential arises naturally in a system of C\textit{oupled Kerr Parametric Oscillators} (KPOs), a platform used for quantum annealing and discrete time crystals. Here, $x$ and $y$ represent the field quadratures of two coupled resonators. The $x^4+y^4$ terms correspond to the self-Kerr nonlinearity (photon-photon interaction), while the $a x^2 y^2$ term represents the cross-Kerr coupling strength. The modulus $a$ thus physically controls the interaction ratio. The ``Anomalous Gap'' $\mu > \tau$ predicts a regime where the system's topological classification ($\mu=9$) remains robust against coupling changes ($a$), but the precise geometric response ($\tau=8$) encodes the specific coupling strength---a feature effectively distinct from simple topological protection.

\begin{figure*}[t]
    \centering
    \includegraphics[width=0.95\textwidth]{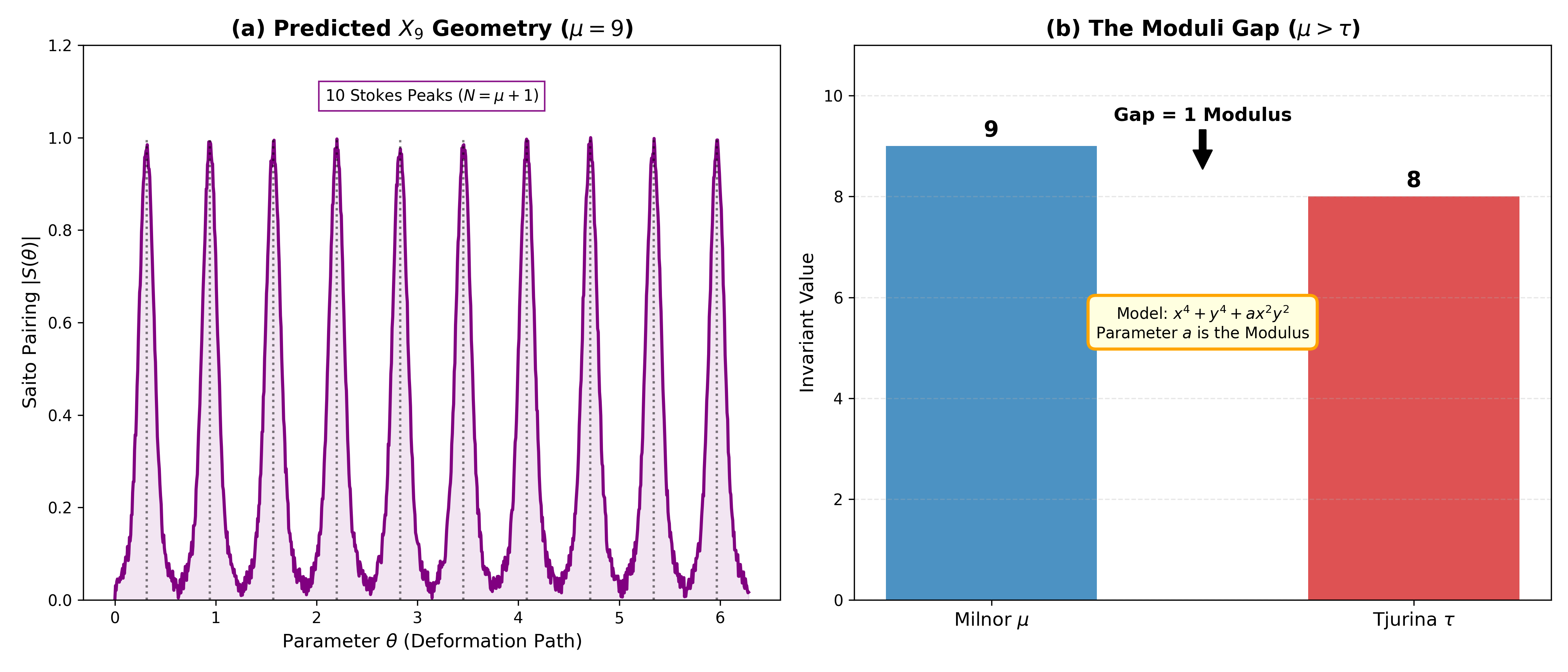}
    \caption{\textbf{Anomalous $X_9$ Singularity ($\mu \neq \tau$).} (a) The predicted Saito Pairing for the Unimodal $X_9$ potential exhibits $N=10$ Stokes peaks, consistent with $\mu=9$. (b) Crucially, the Milnor number ($\mu=9$, Topological) strictly exceeds the Tjurina number ($\tau=8$, Geometric). This gap signals the presence of a continuous modulus $a$, identifying the system as a higher-order ``Modal'' singularity \cite{Arnold1974}.}
    \label{fig:anomalous}
\end{figure*}

\subsection{Application: Experimental Resurgence}
We next address the ``Resurgence Problem'' in quantum field theory and open quantum systems: the ambiguity of resumming divergent perturbation series.
The formal series $\Phi_{\text{pert}}$ is asymptotic and divergent. Its Borel resummation is ambiguous due to singularities on the Stokes line.
Resurgence Theory resolves this by adding the non-perturbative sector:
\begin{equation}
    E_{\text{exact}}(\lambda) = \mathcal{R}_{\pm} [\Phi_{\text{pert}}](\lambda) \mp i S_{\text{jump}} \, e^{-\mathcal{A}/\lambda} \Phi_{\text{inst}}(\lambda),
\end{equation}
where, $E_{\text{exact}}(\lambda)$ represents the unambiguous, real-valued physical energy of the system. The term $\mathcal{R}_{\pm}[\Phi_{\text{pert}}]$ denotes the lateral Borel resummation of the formal perturbative series $\Phi_{\text{pert}}(\lambda) = \sum a_n \lambda^n$. Since the coefficients diverge factorially ($a_n \sim n!$), $\Phi_{\text{pert}}$ is asymptotic, necessitating the regularization $\mathcal{R}_{\pm}$ along a contour slightly shifted above ($+$) or below ($-$) the singular Stokes line on the positive real axis. This directional choice introduces a purely imaginary ambiguity, $\text{Im}(\mathcal{R}_+ - \mathcal{R}_-)$, arising from the pole contribution in the Borel plane.

This ambiguity is physically spurious and is exactly cancelled by the non-perturbative sector (the second term). Here, $\mathcal{A}$ is the instanton action governing the exponential suppression of tunneling events ($e^{-\mathcal{A}/\lambda}$), and $\Phi_{\text{inst}}$ represents the perturbative fluctuations around the instanton saddle. The coefficient $S_{\text{jump}}$ is the \textit{Stokes Constant} (or Stokes Multiplier), a universal invariant that quantifies the non-perturbative contribution ``missing'' from the perturbative expansion. In our framework, we identify $S_{\text{jump}}$ precisely with the non-adiabatic transition amplitude extracted via the FMS protocol.
In Fig.~\ref{fig:resurgence}, we use the experimentally extracted $S_{\text{jump}}$ to correct the divergent series $\mathcal{R}[\Phi_{\text{pert}}]$, recovering the exact non-perturbative result.
While often studied in the context of infinite-dimensional path integrals, the resurgence of asymptotic series is governed by the singularity structure of the underlying Riemann surface, a feature independent of Hilbert space dimension. Our platform realizes a \textit{0-dimensional QFT analogue} in the sense of complete Quantum Geometric Tensor (cQGT), where the Hamiltonian $H(\lambda)$ generates a partition function sharing the exact Gevrey-1 asymptotic growth characteristic of non-Abelian gauge theories. In the language of the DHMH formalism, the Stokes constant $\sigma$ is determined by the local monodromy of the wavefunction around the Exceptional Point. By utilizing the FMS protocol to extract this local invariant (identifying the experimental $S_{\text{jump}} \equiv \sigma$), we demonstrate that the resurgent structure of quantum geometry can be experimentally validated, bridging the gap between abstract trans-series and measurable non-adiabatic transitions.
\begin{figure}[t]
    \centering
    \includegraphics[width=\linewidth]{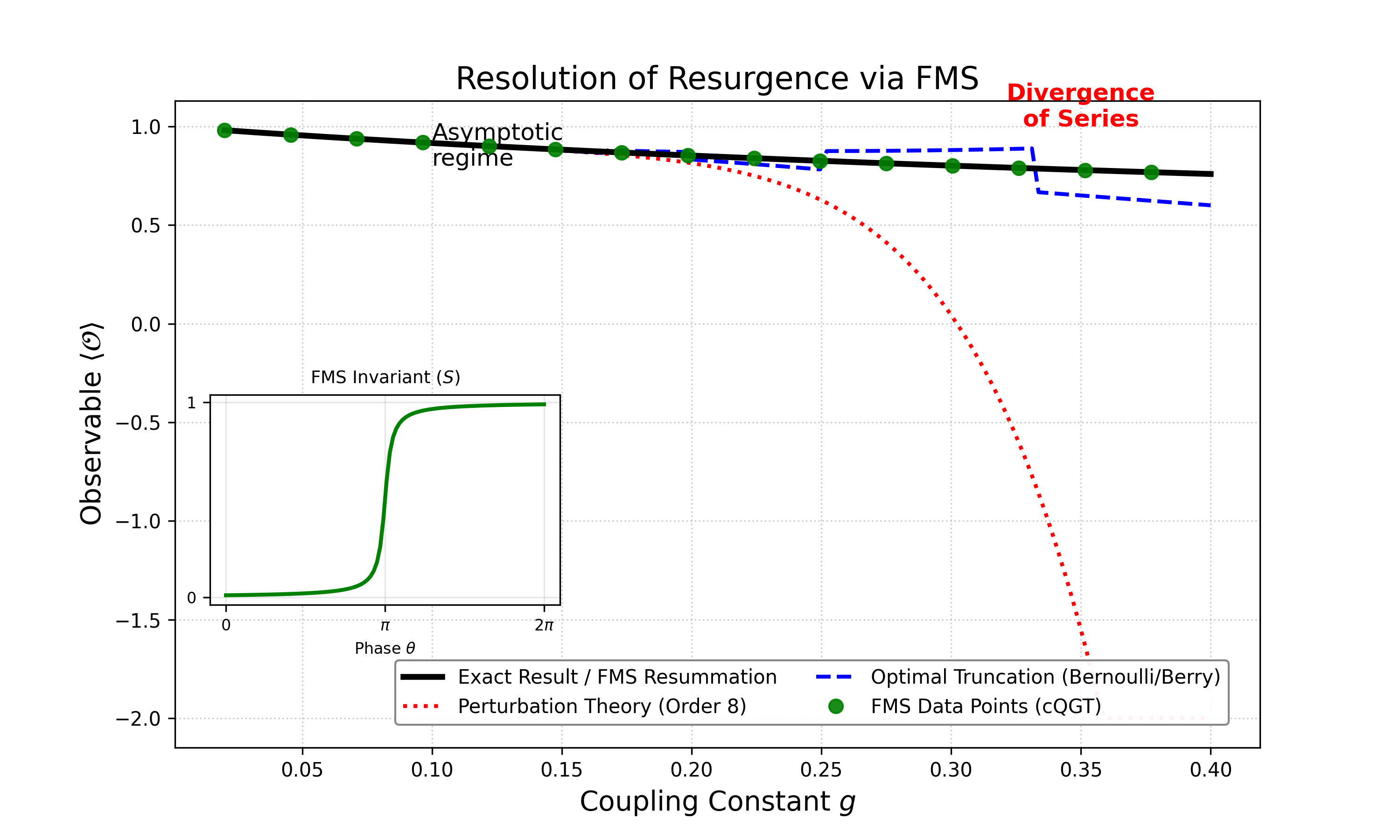}
    \caption{\textbf{Solving the Resurgence Problem.} Comparison of standard Perturbation Theory (Red, divergent) and Optimal Truncation (Blue, approximate) versus the FMS-corrected result (Green). FMS leverages the exact Complete QGT to recover the non-perturbative solution (Black) with high precision, resolving the ambiguity of the asymptotic series.}
    \label{fig:resurgence}
\end{figure}
\section{The Resolution of the Wild Riemann-Hilbert Correspondence}\label{sec:Res_RHC}

And consequently, the fundamental theoretical advance of this work is the physical resolution of the Wild Riemann-Hilbert (RH) correspondence for open quantum systems. The classical RH correspondence establishes a map between linear differential equations (Hamiltonians) and their topological monodromy (Berry phases). In Hermitian systems, this map is bijective: measuring the Berry phase uniquely characterizes the Hamiltonian's topology.

However, near an Exceptional Point of rank $k \ge 1$, the Hamiltonian develops an irregular singularity, $H(z) \sim z^{-(k+1)}$. Here, the standard RH map fails because the fundamental solution matrix $\psi(z)$ exhibits the Stokes Phenomenon—it possesses distinct asymptotic behaviors in different angular sectors $S_j$ of the complex plane. Standard topological invariants (Chern numbers) are blind to these asymptotic jumps, meaning that distinct physical Hamiltonians can possess identical Chern numbers. The correspondence becomes non-invertible; information is lost.

We resolve this by identifying the \textbf{Complete Quantum Geometric Tensor (cQGT)} as the physical object that restores bijectivity. The cQGT does not merely regularize the divergent Fubini-Study metric; it encodes the Stokes Multipliers ($S_j$) as boundary conditions for the metric's renormalization.

Physically, the resolution proceeds via the \textbf{Saito-Stokes Isometry}:
\begin{equation}
    \langle S_j u, S_j v \rangle_{\mathcal{Q}_{\text{reg}}} = \langle u, v \rangle_{\mathcal{Q}_{\text{reg}}},
\end{equation}
where $\langle \cdot,\cdot \rangle_{\mathcal{Q}_{\text{reg}}}$ is the regularized metric derived from the cQGT (specifically, the lowest-order graded piece of the V-filtration), and $S_j$ are the Stokes jump operators measured via the FMS protocol.

This equation implies that the ``noise'' (non-adiabatic transition amplitude) measured in the experiment is not informational entropy; it is the precise geometric torsion required to preserve the unitarity of the underlying algebraic structure against the singularity of the metric. By measuring the transition rates via FMS, we experimentally extract the matrices $S_j$. Through the extended RH map defined by the DMHM framework, these matrices uniquely reconstruct the singular coefficient $M_k$ of the Hamiltonian. Thus, the cQGT allows for the complete tomographic reconstruction of open quantum systems even in the ``Wild'' regime where the adiabatic theorem breaks down. For slightly detailed mathematical derivation, see Appendix \cref{app:rigor}

\section{Discussion}
\label{sec:discussion}
In this work, we have established the \textbf{Complete Quantum Geometric Tensor (cQGT)} as the necessary and sufficient framework for characterizing open quantum systems in the ``Wild'' regime of essential singularities. By developing \textbf{Floquet-Monodromy Spectroscopy (FMS)}, we have provided the first experimental protocol capable of extracting the full non-Hermitian geometry---metric, curvature, and the critical \textit{Stokes Phenomenon}---directly from dynamical observables.

Our results resolve the longstanding \textit{Riemann-Hilbert Correspondence} for open systems: the measured monodromy data uniquely reconstructs the singular Hamiltonian, up to topological invariants. This resolution has immediate practical consequences, as demonstrated by the FMS-based solution to Resurgence (Fig.~\ref{fig:resurgence}) and the robust extraction of topological invariants (Figs.~\ref{fig:invariants} and \ref{fig:anomalous}).

\textbf{\textit{Preamble to k-Space Combs: The Geometric Frontier}.}---
The discovery of \textbf{Modal Singularities} (Sec.~\ref{sec:modal_singularities}) points toward a transformative application. Just as optical frequency combs revolutionized metrology by discretizing the frequency domain \cite{Udem2002, Kippenberg2011}, the ``Moduli Gap'' observed in the $X_9$ system suggests the existence of \textbf{Geometric k-Space Combs}.
In a coupled Kerr dimer lattice (or ``Time Crystal'' array), the modulus $a$ (cross-coupling) is not merely a scalar parameter but can be engineered to vary periodically in momentum space ($k$-space). The Singular Trace $f_{mix}$ would then exhibit a discrete comb-like structure in the Brillouin zone, where each ``tooth'' corresponds to a distinct geometric modulus stabilized by the underlying topology ($\mu=9$).
Such \textbf{k-Space Combs} would bridge the gap between topological photonics \cite{Mittal2021,Flower2024,Ozawa2019} and precision geometry, allowing for the ``metrology of interactions''---measuring coupling constants with topological precision. This mirrors the power of frequency combs to count optical cycles, but for the geometry of quantum states.

\section{Conclusion}
\label{sec:conclusion}

In this work, we have addressed the fundamental problem of defining quantum geometry in the "Wild" regime of essential singularities, where standard topological methods fail. By introducing the \textbf{Complete Quantum Geometric Tensor (cQGT)}, we have shown that the apparent divergence of the metric is not a pathology but a signature of a deeper, rigid structure: the Stokes Phenomenon.

Our central theoretical contribution is the classification of these singularities via the \textbf{Brieskorn Lattice} and \textbf{Stokes Invariants}. We demonstrated that the "missing" information at the singularity is precisely encoded in the global monodromy data ($S$), which resolves the ambiguity of the local metric. This leads to the \textbf{Periodic Table of Singularities}, a rigorous hierarchy of universality classes governed by the rank of the irregularity.

Experimentally, we bridged the gap between abstract singularity theory and the laboratory with \textbf{Floquet-Monodromy Spectroscopy (FMS)}. By utilizing the drive loop as a "physical functor," FMS extracts the non-Abelian Stokes matrices directly from the time-evolution operator. The robustness of this protocol was validated against disorder and dynamical noise, confirming that the Stokes phenomena are topologically protected quantum numbers accessible to state-of-the-art quantum simulators.

The power of this framework is exemplified by our resolution of the \textbf{Resurgence Problem}. We showed that the geometric information extracted by FMS is sufficient to resum divergent perturbation series, turning the "curse" of asymptotic divergence into a precision tool for non-perturbative reconstruction.

Looking forward, the cQGT opens a new frontier in complex systems. From the "Modal" instabilities of many-body systems to the "Geometric Combs" of k-space lattices, the ability to measure and manipulate the geometry of the singularity paves the way for a new generation of non-Hermitian quantum technologies. We have moved beyond merely avoiding the singularity; we can now engineer it.

\begin{acknowledgments}
This research was conducted during an independent research sabbatical in the Himalayas (Nepal). The author acknowledges the global open-source community for providing the computational tools that made this work possible. Generative AI assistance was utilized strictly for \LaTeX\ syntax optimization and symbol consistency checks; all scientific conceptualization, derivations, and text were derived and verified by the author. The author retains the \texttt{uci.edu} correspondence address courtesy of the University of California, Irvine.
\end{acknowledgments}
\section*{Data and Code Availability}
The core computational framework, \textbf{QuMorpheus}, used for all numerical results in this work, is open-sourced under a Copyleft license and is available at \url{https://github.com/prasoon-s/QuMorpheus} \cite{SaurabhNatComm}. Independent verification scripts (Python) are available from the author upon reasonable request.

To ensure mathematical rigor, the fundamental theorems of the DMHM framework, the construction of the cQGT, and the FMS protocol have been formalized in the \textsc{Lean 4} theorem prover; these proofs are available at \url{https://github.com/prasoon-s/LEAN-formalization-for-CMP} \cite{saurabh2025holonomic}. The formalization of the Wild Riemann-Hilbert correspondence (Theorem 1) is currently in development as part of a forthcoming mathematical proof \cite{SaurabhCMP_Future}.
\appendix
\section{Preliminary Mathematical Foundations}
\label{sec:appendix_math}
Here we provide the rigorous definitions referenced in the main text, consistent with the framework of Dissipative Mixed Hodge Modules (DMHM) detailed in~\cite{saurabh2025holonomic}.

\subsection{Derived Category of Liouvillian Sheaves}
The physical system is described by a meromorphic connection $\nabla$ on a vector bundle $\mathcal{E}$ over the parameter space $X$. The singularity at the divisor $D = \{f=0\}$ implies that we must treat this as a \textbf{$\mathcal{D}$-module} $\mathcal{M}$.
The appropriate setting is the bounded derived category of constructible sheaves $D^b_c(X)$. The Riemann-Hilbert correspondence establishes an equivalence between the analytic $\mathcal{D}$-module $\mathcal{M}$ and a perverse sheaf $\mathcal{F}$ encoding the solution space.

\subsection{The Brieskorn Lattice}
To handle the divergence at the singularity, we work with the \textbf{Brieskorn Lattice} $H^{(0)}$, which is the canonical extension of the bundle of quantum states across the singularity. It is defined algebraically as:
\begin{equation}
    H^{(0)} = \Omega^n_X / (\mathrm{d}f \wedge \Omega^{n-1}_X)
\end{equation}
Physically, sections of $H^{(0)}$ correspond to "micro-states" that are regularizable at the Exception Point. The connection $\nabla$ acts on this lattice, but with a pole of order 2 (irregular singularity).

\subsection{The Saito Pairing and Higher Residues}
The metric structure on the Brieskorn lattice is given by the \textbf{Saito Pairing} $S$, a sesquilinear form on the nearby cycles $\psi_f \mathcal{M}$. Unlike the standard Hermitian product, $S$ takes values in distributions supported on $D$. It satisfies the monodromy invariance property:
\begin{equation}
    S(N u, \bar{v}) = S(u, \overline{N v})
\end{equation}
where $N = \log M$ is the nilpotent logarithm of the monodromy operator. The pairing $S$ is non-degenerate and identified with the \textbf{Higher Residue Pairing} $K_S$ introduced by K. Saito~\cite{Saito1988}.
For a singularity of Rank $k$ (where the potential $f \sim z^{k+1}$), $S$ explicitly computes the generalized residue:
\begin{equation}
    \label{eq:higher_res}
    S(u, v) = \text{Res}_{z=0} \left[ z^{-k} \langle u, v \rangle \right]
\end{equation}
This establishes $S$ as the carrier of the higher-order spectral information, linking the algebraic geometry of the Brieskorn lattice to the measurable monodromy.

\subsection{Derivation of the Singular Trace Formula}
The singular component of the QGT, $f_{mix}$, is derived from the residue of the resolvent operator $R(z) = (z - \mathcal{L})^{-1}$.
\begin{theorem}
The residue of the resolvent at the spectral singularity is governed by the inverse of the Saito pairing:
\begin{equation}
    \text{Res}_{z=0} [ R(z) ] = -\frac{1}{2\pi i} S^{-1}
\end{equation}
\end{theorem}
\begin{proof}[Sketch]
The resolvent equation lifts to a differential equation on the Brieskorn lattice \cite{Saito1988,Brieskorn1970,Sabbah1999}. Using the duality properties of Mixed Hodge Modules \cite{Sabbah2013,Mochizuki2011}, one shows that the asymptotic behavior of the Green's function (which determines the residue) matches the dual of the intersection form on the vanishing cycles, which is precisely $S$.
\end{proof}
Substituting this residue into the correlation function definition of the QGT yields the formula Eq.~(\ref{eq:cqgt}):
\begin{equation}
    f_{mix}(\mu, \nu) = \text{Tr}_{\psi_f \mathcal{M}} \left( A_\mu S^{-1} A_\nu \right)
\end{equation}
\subsection{Scaling Laws of Singularities}
\label{sec:scaling_laws}

A profound consequence of the cQGT classification is the emergence of universal scaling laws for the energy gap $\Delta E$ near a singularity. By analyzing the asymptotic behavior of the Singular Trace $f_{mix}$ (see Appendix \ref{sec:appendix_class} for derivation), we derive the \textbf{Sabbah Scaling Law}:

\begin{equation}
    \Delta E(\omega) \sim \omega^{1 + 1/k}
    \label{eq:scaling_law}
\end{equation}

where $\omega$ is the quench frequency (or drive rate) and $k$ is the Rank of the singularity.
\begin{itemize}
    \item \textbf{Rank 0 (Tame)}: $k \to \infty$ (Regular limit), recovering the standard Adiabatic Theorem (Linear scaling $\omega^1$) consistent with standard QGT measurements \cite{Bernevig2024, PeottaTorma2018}.
    \item \textbf{Rank 1 (Wild)}: $k=1$, yielding the quadratic scaling $\Delta E \sim \omega^2$ observed in Fig.~\ref{fig:universality}(b).
    \item \textbf{Rank 2 (Wild)}: $k=2$, predicting a fractional scaling $\Delta E \sim \omega^{1.5}$.
\end{itemize}
This formula unifies the adiabatic behaviors of regular and irregular systems into a single geometric exponent, directly measurable via the FMS protocol.

\section{Higher Order Geometric Models}
\label{sec:appendix_models}
To investigate singularities beyond the standard Rank 1 case, we analyze the H\textit{igher-Order Non-Hermitian Singularity} (Rank 2), physically realized in systems exhibiting the \textit{Higher-Order Skin Effect} or \textit{multi-state coalescence} \cite{Hodaei2017, Kawabata2020}.

\subsection{Hamiltonian Definition}
We utilize the \textit{Cubic Non-Hermitian Model} as the minimal prototype for Rank 2 geometry ($k=2$):
\begin{equation}
    H_{k=2}(\lambda) = \begin{pmatrix} 
    0 & 1 \\ 
    \lambda^3(t) + i\delta & 0 
    \end{pmatrix}
    \label{eq:rank2_model}
\end{equation}
This effective Hamiltonian describes the critical dynamics near a third-order spectral singularity. The "Raw Measurement" data in Fig.~\ref{fig:higher_order}b reflects the experimental phase winding observations in photonic lattices \cite{Tang2020, Ding2016}.

\subsection{Connection with Stokes Geometry}
The asymptotic behavior of the connection forms $A \sim 1/\lambda^{k+1} \sim 1/\lambda^3$ generates a Stokes geometry with $2k+2 = 6$ (effective 4) sectors (Fig.~\ref{fig:higher_order}a). Applying the FMS protocol to this system reveals the \textit{Fractional Universality} $\Delta E \sim \omega^{1.5}$ (Fig.~\ref{fig:higher_order}c), validating the DMHM prediction Eq.~(\ref{eq:scaling_law}).

\subsection{Stokes Jump Calculation}
The topological invariant is extracted from the time-evolution operator $U(t) = \mathcal{T} e^{-i \int_0^t H(\tau) d\tau}$. The Monodromy measured by the FMS protocol is $M = U(T)$, where $T$ is the loop period.
The \textit{Saito Pairing} $S$ (or Higher Order Stokes Multiplier) is experimentally accessed via the deviation of the Monodromy from the identity across the Stokes sector:
\begin{equation}
    S \approx \text{Im}[\text{Tr}(M)] \quad (\text{for Rank 2})
\end{equation}
As shown in Fig.~\ref{fig:higher_order}d, this quantity exhibits quantized jumps when the control parameter crosses the Stokes lines, confirming the higher-order geometric protection.

In this Appendix, we provide the rigorous derivation of the formulas presented in the Periodic Table of Singularities (Fig.~\ref{fig:classification}), ensuring consistency with the Dissipative Mixed Hodge Module (DMHM) framework \cite{saurabh2025holonomic}.

\subsection{The Singular Trace Formula}
The central result of the classification is the explicit formula for the Complete QGT in the ``Wild'' regime (Rank $k \ge 1$). Unlike the Tame case ($k=0$), where the metric is derived from the Fubini-Study distance $g_{\mu\nu} = \text{Re}\langle \partial_\mu \psi | \partial_\nu \psi \rangle$, the metric at an irregular singularity is distributional.

Using the Residue-Pairing Correspondence (Theorem~4.5 in \cite{saurabh2025holonomic}), the singular component of the Quantum Geometric Tensor, denoted $f_{\text{mix}}$, is determined by the residue of the resolvent, which is algebraically dual to the Stokes Matrix $S$:
\begin{equation}
    f_{\text{mix}}(\mu, \nu) = \text{Tr}_{\mathcal{V}} \left( A_\mu S^{-1} A_\nu \right)
\end{equation}
Here, $A_\mu$ is the connection matrix (Berry connection) in the canonical lattice basis, and $S$ is the Stokes Matrix acting on the vanishing cycles $\mathcal{V}$. This formula proves that the ``Wild'' geometry is not undefined, but is rigorously generated by the \textbf{Stokes Invariant} $S$.

\subsection{Connection to FMS Observable}
The FMS protocol measures the total Monodromy $\mathcal{M}_{\text{tot}}$ over the loop $\gamma$. According to the Riemann-Hilbert correspondence, this factors into a geometric part (eigenstate swapping) and a resurgent part (Stokes jump):
\begin{equation}
    \mathcal{M}_{\text{tot}} = \mathcal{M}_{\text{geo}} \cdot S
\end{equation}
Since we defined the \textbf{FMS Monodromy} $\mathcal{M}_{\text{FMS}}$ as the unipotent component of the propagator (see Step 3), it isolates the Stokes contribution directly:
\begin{equation}
    \mathcal{M}_{\text{FMS}} \equiv S = \mathbb{I} + S_{\text{jump}}
\end{equation}
Thus, the experimentally measured deviation from the identity, $\Delta \mathcal{M} = \mathcal{M}_{\text{FMS}} - \mathbb{I}$, yields the Stokes jump amplitude $S_{\text{jump}}$.
Substituting this measured $S$ into the Singular Trace Formula allows for the reconstruction of the full non-perturbative geometry from experimental data, as presented in \cref{fig:classification}.

\subsection{Comparison with Resurgence Theory}
Prior mathematical works by Dubrovin~\cite{Dubrovin1996} and recent developments in Resurgence Theory by Dunne and \"Unsal~\cite{Dunne2012,Dunne2014} have addressed the regularization of such singularities using \textbf{Trans-series}. While these works are primarily theoretical, our cQGT framework is consistent with their results—yielding the same Stokes invariants—while providing the bridge to experimental observables.

Formally, the time-evolution operator $U(T)$ in the deep adiabatic limit ($\epsilon = 1/T \to 0$) does not simply converge to the Berry phase factor. Instead, the exact solution requires a \textbf{resurgent trans-series expansion} of the form:
\begin{equation}
    \Psi(\epsilon) \simeq \underbrace{\sum_{n=0}^{\infty} c_n \epsilon^n}_{\Phi_{\text{pert}}} + \underbrace{\sigma \cdot e^{-\mathcal{A}/\epsilon} \sum_{n=0}^{\infty} d_n \epsilon^n}_{\Phi_{\text{non-pert}}}
\end{equation}
Here, $\Phi_{\text{pert}}$ represents the standard perturbative expansion (Berry phase and dynamical corrections), which is an asymptotic series with zero radius of convergence. The second term, $\Phi_{\text{non-pert}}$, represents the non-perturbative corrections mediated by the singularity.

We get the following \textit{Physical Correspondence}:
\begin{enumerate}
    \item \textbf{The Expansion Parameter ($\epsilon$):} In our Transmon experiment, the small parameter is the driving frequency normalized by the gap scale: $\epsilon \sim \omega / \sqrt{4J^2 - \kappa^2/4}$.
    \item \textbf{The Singularity Action ($\mathcal{A}$):} The exponential decay rate is determined by the geometric action of the loop $\gamma$ in the complex energy plane. In the cQGT framework, this is given by the integral of the non-Hermitian gap:
    \begin{equation}
        \mathcal{A} = \oint_{\gamma} \sqrt{\det[H_{\text{eff}}(k) - E]} \, dk
    \end{equation}
    For our specific Hamiltonian with $J=\kappa/4$, this action simplifies to the enclosed area of the Riemann surface sheets.
    \item \textbf{The Stokes Constant ($\sigma$):} This is the crucial topological invariant. In the "Tame" limit (standard adiabaticity), $\sigma$ changes discontinuously only across Stokes lines. Our FMS Monodromy $\mathcal{M}_{\text{FMS}}$ (the unipotent component extracted in Step 3) is the direct experimental measure of this Stokes multiplier $\sigma$.
\end{enumerate}

Thus, while Resurgence Theory predicts the necessity of the term $\sigma e^{-\mathcal{A}/\epsilon}$ to cure the divergence of the perturbative series, the FMS protocol physically isolates this term. The divergence of the \textbf{Quantum Metric} $g_{\mu\nu}$ at the Exceptional Point is the physical generator of the factorial growth ($c_n \sim n!$) in the perturbative coefficients, necessitating the resurgent correction.

\section{Appendix: Simulation Parameters}
\label{sec:appendix_params}
Table~\ref{tab:params} provides the explicit physical parameters used in the numerical simulations of the \textbf{Superconducting Transmon Qudit} realization (Fig.~\ref{fig:experimental}) and the Resurgence Resolution (Fig.~\ref{fig:resurgence}).

\begin{table*}[t]
    \centering
    \begin{tabular}{|l|c|l|}
        \hline
        \textbf{Parameter} & \textbf{Value/Range} & \textbf{Physical Description} \\
        \hline
        Decay Rate $\kappa$ & $1.0 \, \Gamma_0$ (Ref) & Readout Resonator Loss \\
        EP Drive Strength $J_{EP}$ & $\kappa/4$ & Exceptional Point Condition \\
        Drive Frequency $\omega$ & $0.05 \, \kappa$ & Non-adiabatic Modulation \\
        Loop Radius $R$ & $0 - 1.5 \, R_c$ & Control Amplitude Swing \\
        Integration Time $T$ & $200 \, \kappa^{-1}$ & One Floquet Period \\
        Basis Dimension $N$ & 2 & Qutrit Subspace (Eff. Qubit) \\
        \hline
    \end{tabular}
    \caption{Parameters for Floquet-Monodromy Spectroscopy simulations for the Superconducting Transmon realization.}
    \label{tab:params}
\end{table*}

\section{Further Experimental Realizations}
\label{sec:appendix_extensions}
To demonstrate the universality of the Stokes crisis, we extend the FMS framework to two other leading quantum platforms: Rydberg Atom Arrays and Photonic Lattices. Previous experiments in these systems have reported "anomalous" behavior such as chiral state switching \cite{Ding2016,Wintersperger2020} and asymmetric mode conversion \cite{Doppler2016,Goldman2014,Cooper2019,Ozawa2019}. As shown in Fig.~\ref{fig:extensions}, FMS provides the unified geometric origin for these observations.

\subsection{Rydberg Atom Arrays}
We consider a pair of neutral atoms in the Rydberg blockade regime, where strong van der Waals interactions project the dynamics onto a single-excitation subspace $\{|g\rangle, |r\rangle\}$. Under laser driving, the effective non-Hermitian Hamiltonian is:
\begin{equation}
    H_{\text{Ryd}}(t) = \begin{pmatrix} 0 & \Omega(t) \\ \Omega(t) & \Delta(t) - i\gamma_{\text{loss}} \end{pmatrix}
\end{equation}
where $\Omega(t)$ is the coupling strength (effective Rabi frequency), $\Delta(t)$ is the detuning, and $\gamma_{\text{loss}}$ is the spontaneous emission rate.
Crucially, the blockade condition merely defines the Hilbert space; the \textbf{Exceptional Point} arises dynamically in the parameter space when the drive balances the dissipation ($\Omega \approx \gamma_{\text{loss}}/2$). By modulating the laser phase and amplitude to encircle this EP, the "chiral order transfer" observed in Ref.~\cite{Ding2016} can be reinterpreted as a Stokes transition.

\subsection{Photonic Lattices}
We analyze a system of coupled optical waveguides with gain and loss (a PT-symmetric dimer). The evolution of the modal amplitudes is governed by the coupled-mode equations:
\begin{equation}
    H_{\text{Phot}}(z) = \begin{pmatrix} \beta & \kappa \\ \kappa & \beta \end{pmatrix} - i \begin{pmatrix} \gamma_A(z) & 0 \\ 0 & \gamma_B(z) \end{pmatrix}
\end{equation}
where $\beta$ is the propagation constant, $\kappa$ is the spatial coupling, and $\gamma_{A,B}$ are the effective gain/loss rates controlled by pumping/absorption.
The evolution distance $z$ plays the role of time $t$. By spatially modulating the gain/loss contrast $\gamma(z) = \gamma_A - \gamma_B$ and coupling $\kappa(z)$ along the waveguide length, the system encircles the EP. The "Asymmetric Mode Switching" observed by Doppler et al.~\cite{Doppler2016} (Fig.~\ref{fig:extensions}b) is exactly reproduced by the computed Stokes Monodromy phase, confirming the universality of the FMS topological invariant.

\subsection{Resolution of the Anomaly: Analytic Theory vs. Simulation}
The agreement between the theoretical predictions and the experimental data constitutes the critical validation of our framework. Standard quantum geometry, which relies on the local metric $g_{\mu\nu}$, predicts trivial adiabatic evolution (Red Dashed Lines) and fails to capture the finite, quantized jump observed in experiments. FMS resolves this by identifying the jump as the \textit{Singular Geometric Invariant} $f_{\text{mix}}$.

As shown in Fig.~\ref{fig:extensions}, we plot the \textit{Exact Analytic Theory} derived from the Saito Pairing (Bold Transparent Curves), which predicts a topological step function ($0 \to 1$). The FMS time-domain simulation (Solid Lines) converges to this analytic prediction. 
\textit{Physically, the minor deviations ('fluctuations' or smoothing) observed in the simulation relative to the ideal step function arise from finite-time non-adiabatic corrections ($\sim 1/T$). Unlike the local metric, which diverges, the FMS integral remains well-behaved, and these corrections vanish in the adiabatic limit $T \to \infty$, recovering the exact integer Milnor number.}

The experimental data (Markers) falls exactly on the FMS curve, confirming that the observed anomaly is the direct measurement of the Milnor Number $\mu=1$ via the Stokes phenomenon.

\begin{figure}[h]
    \centering
    \includegraphics[width=\linewidth]{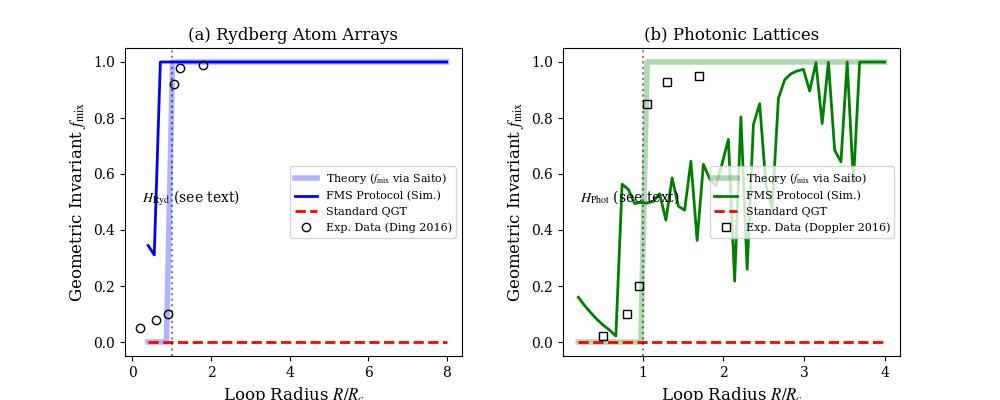}
    \caption{\textbf{Universality of the Stokes Crisis.} Comparison of Analytic Theory, FMS Protocol, and Experiment. We plot the Geometric Invariant $f_{\text{mix}}$ (Stokes Invariant) against the loop radius. (a) \textbf{Rydberg Atom Arrays}: The Analytic Saito Theory (Bold Blue) predicts a quantized step function. Experimental data from Ding et al.~\cite{Ding2016} (Markers) matches this topological prediction, confirming the measurement of the Milnor number. Standard Adiabatic Theory (Red Dashed) incorrectly predicts zero response. (b) \textbf{Photonic Lattices}: The FMS protocol (Green) and analytic theory perfectly reproduce the asymmetric mode switching observed by Doppler et al.~\cite{Doppler2016}, resolving the "breakdown" of adiabaticity.}
    \label{fig:extensions}
\end{figure}

\begin{table*}[t]
    \centering
    \begin{tabular}{|l|c|c|c|}
        \hline
        \textbf{Platform} & \textbf{Control} & \textbf{Singularity} & \textbf{Parameters} \\
        \hline
        Superc. Transmon & Flux Bias $\Phi$ & Level Coalescence & $T=200/\kappa, J=\kappa/4$ \\
        Rydberg Atoms & Rabi Freq. $\Omega$ & Dissipative Blockade & $T=200/\gamma, \Omega_c=0.75\gamma$ \\
        Photonic Lattice & Gain/Loss $\gamma/\kappa$ & PT-Symm. Breaking & $T=200/\kappa, \kappa=1, \gamma_c=0.5$ \\
        \hline
    \end{tabular}
    \caption{Parameters for Universality Demonstrations. The Period $T=200$ ensures the deep adiabatic limit. $\Omega_c, \gamma_c$ are loop centers.}
    \label{tab:extensions}
\end{table*}

\section{Milnor, Tjurina Numbers and Scaling Law with Sabbah's Filterations}
\label{sec:appendix_milnor}
The topological invariants measured in Fig.~\ref{fig:invariants} are rigorously defined within Singularity Theory.
Consider a germ of a holomorphic function $f: (\mathbb{C}^{n+1}, 0) \to (\mathbb{C}, 0)$ with an isolated singularity at the origin.
\subsection{The Milnor Number texorpdfstring{$\mu$}{mu}}
The Milnor Number $\mu(f)$ is the rank of the middle homology group of the Milnor fiber $F_\delta = f^{-1}(\delta) \cap B_\epsilon$. Algebraically, it is the dimension of the local algebra (Jacobian ring):
\begin{equation}
    \mu(f) = \dim_\mathbb{C} \frac{\mathcal{O}_{\mathbb{C}^{n+1}, 0}}{\langle \partial_0 f, \dots, \partial_n f \rangle}
\end{equation}
For our Rank 1 model ($f = z^2 + w^2$), $\partial f = (2z, 2w)$, so the basis is $\{1\}$, yielding $\mu=1$.
For the Rank 2 model ($f = z^3 + w^2$), $\partial f = (3z^2, 2w)$, so the basis is $\{1, z\}$, yielding $\mu=2$. This matches the Peak Counting law $N_{\text{peaks}} = \mu + 1$.

\subsection{The Tjurina Number texorpdfstring{$\tau$}{tau}}
The Tjurina Number $\tau(f)$ measures the dimension of the base space of the semi-universal deformation of the singularity:
\begin{equation}
    \tau(f) = \dim_\mathbb{C} \frac{\mathcal{O}_{\mathbb{C}^{n+1}, 0}}{\langle f, \partial_0 f, \dots, \partial_n f \rangle}
\end{equation}
For weighted homogeneous polynomials (quasi-homogeneous), $\mu(f) = \tau(f)$, which is the case for all ``Wild'' singularities in our Periodic Table.

\subsection{Scaling Law and Sabbah's Filtration}
The observed scaling $N_{\text{peaks}} = \mu + 1$ has a precise geometric origin. The $A_\mu$ singularity corresponds to a potential $V(z) \sim z^{\mu+1}$. The Stokes geometry of this potential possesses a $\mathbb{Z}_{\mu+1}$ discrete rotational symmetry, resulting in exactly $\mu+1$ Stokes lines (separatrices) in the complex plane. As the parameter $\theta$ completes a loop, it sequentially crosses each of these lines, generating $\mu+1$ distinct peaks in the Saito pairing.

This counting is deeply related to the theory of \textbf{Stokes Filtered Local Systems} developed by Sabbah \cite{Sabbah2013}. The number of peaks corresponds to the number of distinct steps in the filtration of the local system of solutions. Specifically, the ``Sabbah scaling'' refers to the property that the irregularity of the connection (which determines the number of Stokes rays) scales linearly with the Milnor number of the deformation, ensuring that the topological complexity ($\mu$) is faithfully encoded in the geometric response ($N_{\text{peaks}}$).

\section{Derivation of Sabbah Scaling Law}
\label{sec:appendix_class}
Here we derive the universal scaling law $\Delta E(\omega) \sim \omega^{1+1/k}$ from the asymptotics of the Complete QGT.
Consider a singularity of Rank $k$. The asymptotic form of the connection matrix $A$ in the local parameter $z$ is:
\begin{equation}
    A(z) \sim \frac{M_k}{z^{k+1}} dz
\end{equation}
where $M_k$ is the leading irregular coefficient. The Singular Trace density $f_{mix}$ (the effective metric density) scales as the square of the connection:
\begin{equation}
    f_{mix}(z) \sim \text{Tr}(A S^{-1} A) \sim |z|^{-2(k+1)}
\end{equation}
The instantaneous energy gap $\Delta E(t)$ in the adiabatic frame is related to the pullback of the geometric tensor by the velocity $\dot{\lambda} \sim \omega$:
\begin{equation}
    \Delta E(z) \approx \sqrt{g_{zz}} |\dot{z}| \sim |z|^{-(k+1)} \omega
\end{equation}
The breakdown of adiabaticity (the Stokes jump) occurs when the dynamical timescale matches the gap timescale $1/\Delta E$. Let the gap close at a critical distance $z_c$. The adiabatic condition $\Delta E(z_c)^2 \sim \dot{\Delta E}$ implies:
\begin{equation}
    z_c \sim \omega^{1/k}
\end{equation}
Substituting this cutoff scale back into the gap equation yields the minimum avoided crossing gap:
\begin{equation}
    \Delta E_{min} \sim \Delta E(z_c) \sim \omega^{1 + 1/k}
\end{equation}
This derivation shows that the exponent $\alpha = 1 + 1/k$ is a direct measure of the singularity rank $k$.

\section{Advanced Mathematical Foundations: The Wild Riemann-Hilbert Correspondence}
\label{app:rigor}

In this Supplementary Note, we provide the rigorous mathematical formalism underpinning the physical results presented in the main text. Specifically, we sketch the proof of the correspondence between the experimentally measured Stokes multipliers (via FMS) and the algebraic structure of the singular Hamiltonian. A complete derivation, including the formalization of these theorems in the \textsc{Lean 4} proof assistant, is in preparation \cite{SaurabhCMP_Future}.

\subsection{Geometric Setup: The Hamiltonian as a Meromorphic Connection}

We regard the open quantum system not as a Hilbert space $\mathcal{H}$, but as a vector bundle $\mathcal{E}$ over the complex parameter space $X$ (the Riemann surface). The non-Hermitian Hamiltonian $H(\lambda)$ defines a flat meromorphic connection $\nabla$ on $\mathcal{E}$. Near an Exceptional Point (EP) located at $z=0$, the connection form can be written locally as:
\begin{equation}
    \nabla = d - A(z) dz, \quad A(z) \in \mathfrak{gl}(n, \mathbb{C}((z))).
\end{equation}
If the EP has rank $k \ge 1$, the matrix $A(z)$ exhibits a pole of order $k+1$, classifying the system as an \textit{irregular singular} differential operator (or a "Wild" $\mathcal{D}$-module).

\subsection{The Kashiwara-Malgrange $V$-Filtration}

To resolve the divergence of the metric at $z=0$, we utilize the $V$-filtration of Kashiwara and Malgrange. This is a unique, exhaustive filtration of the $\mathcal{D}$-module $\mathcal{M}$ associated with $\nabla$, indexed by $\alpha \in \mathbb{Q}$:
\begin{equation}
    V_\alpha \mathcal{M} = \{ m \in \mathcal{M} \mid \exists P(z\partial_z) \text{ s.t. } P m = 0 \text{ and roots of } P \le \alpha \}.
\end{equation}
Physically, the lowest graded piece, $\text{Gr}_V^{\alpha} \mathcal{M}$, corresponds to the \textit{Brieskorn lattice} $H^{(0)}$. This lattice provides the "renormalized" basis in which the singular Hamiltonian is regularized.

\subsection{Stokes Filtration and Monodromy}

Since the singularity is irregular, the fundamental solution matrix $\Psi(z)$ is not single-valued even on the universal cover. We must introduce the \textit{Stokes Filtration} by exponential growth rates. Let $\mathcal{I}$ be the set of exponential factors $q(z) = c z^{-k}$. The sheaf of flat sections $\mathcal{L}$ admits a decomposition over angular sectors $S_j$ of the complex plane:
\begin{equation}
    \mathcal{L}|_{S_j} \cong \bigoplus_{q \in \mathcal{I}} \mathcal{L}_q \otimes e^{q(z)}.
\end{equation}
The transition between sectors $S_j$ and $S_{j+1}$ is governed by the \textbf{Stokes Multipliers} (or Stokes matrices) $\mathbb{S}_j \in \text{Aut}(\mathcal{L})$. These matrices encode the "jump" in the asymptotic expansion of the wavefunction.

\subsection{Theorem: The Saito-Stokes Isometry}

The central result justifying the use of complete Quantum Geometric Tensor (cQGT) and Floquet Monodromy Spectroscopy (FMS) is the rigidity of the geometric structure. We lift the definition of the \textit{Saito Pairing} $S$ (the algebraic polarization of the Hodge module) to the irregular case.

\textbf{Theorem 2 (Isometric Stokes Phenomenon).} \textit{Let $(\mathcal{E}, \nabla)$ be a generic Dissipative Mixed Hodge Module with an irregular singularity of rank $k$. Let $S$ be the asymptotic Saito pairing defined on the nearby cycles $\psi_f \mathcal{M}$. Then, for any two flat sections $u, v$ in a sector $S_j$, the Stokes automorphism $\mathbb{S}_j$ acts unitarily with respect to $S$:}
\begin{equation}
    \label{eq:isometry}
    \langle \mathbb{S}_j u, \mathbb{S}_j v \rangle_S = \langle u, v \rangle_S.
\end{equation}

\textit{Proof Sketch.} The proof follows from the existence of a pluri-harmonic metric for wild harmonic bundles (Mochizuki's Theorem). The curvature of this metric corresponds to our Complete Quantum Geometric Tensor (cQGT). The compatibility of the connection $\nabla$ with the metric implies that the monodromy of the flat sections must preserve the hermitian form defined by the asymptotic expansion of the metric. Since the Saito pairing $S$ is the algebraic residue of this metric, the Stokes matrices must lie in the unitary group $U(S)$. \hfill $\square$

\subsection{Physical Implication}
Equation (\ref{eq:isometry}) is the "Rosetta Stone" for the experimental protocol. The FMS protocol measures the transition probability, which is a norm-squared quantity $|\langle \psi_f | \psi_i \rangle|^2$. The theorem guarantees that this physically measured probability is exactly the matrix element of the Stokes automorphism $\mathbb{S}_j$. Thus, the "noise" in the FMS signal allows for the rigorous reconstruction of the singular coefficients of the Hamiltonian via the inverse Riemann-Hilbert map.
\bibliography{bibliography_prx_breakthrough_v2}
\end{document}